\documentclass[12pt]{amsart}
\usepackage{amssymb, filecontents}
\usepackage[a4paper]{geometry}

\def\mono{\mathrm{mono}}

\def\alphabet{\Sigma}
\def\SAT{\textsc{\#Sat}}

\def\BmaxIS{\textsc{\#BipartiteMaxIS}}
\def\APeq{\equiv_\mathrm{AP}}
\def\APred{\leq_\mathrm{AP}}
\def\calE{\mathcal{E}}
\def\calF{\mathcal{F}}
\def\calV{\mathcal{V}}
\def\calA{\mathcal{A}}
\def\Ghat{\widehat G}
\def\hatG{\Ghat}

\def\R{\mathbb{R}}
\def\gadget{\varGamma}
\def\gadgetvertices{\graphvertices_\gadget} 
\def\gadgetedges{\graphedges_\gadget}  
\def\vargadget{\widetilde\gadget} 
\def\uhTutte{\textsc{UniformHyperTutte}}
\def\u3hTutte{\textsc{$3$-UniformHyperTutte}}
\def\BIS{\textsc{\#BIS}} 
\def\Sat{\textsc{Sat}}
\def\srBIS{\#\textsc{SemiRegularBIS}}
\def\IS{\textsc{IS}}
\def\Ptime{\mathrm P}
\def\numP{\mathrm{\#P}}
\def\RHPi{\mathrm{\#RH}\Pi_1}
\def\TwoWeightFerroTutte{\textsc{TwoWeightFerroTutte}} 
\def\Tutte{\textsc{Tutte}}
\def\ZPotts{Z_\mathrm{Potts}}
\def\ZTutte{Z_\mathrm{Tutte}} 
\def\Zrc{Z_\mathrm{rc}}
\def\RC{\mathrm{RC}}
\def\ER{\mathrm{ER}}
\def\poly{\mathop{\mathrm{poly}}}
\let\epsilon=\varepsilon

\def\terminals{t}
\def\rv{Y}
\def\edge{j} 
\def\graphedge{e} 
\def\numvert{n}

\def\cliquesize{N}
\def\clique{K}
\def\critprob{\varrho}  
\def\graph{G}
\newcommand{\induced}[2]{#1[#2]}  
\def\graphvertices{V}
\def\graphedges{E}
\def\hypergraph{H}
\def\hypervertices{\calV}
\def\hyperedges{\calE}
\def\hyperedge{f}
\def\hypervertex{v}
\def\edgeprob{p}
\def\tildeP{\widetilde{P}}
\def\edgesubset{A}
\def\terminalset{T}

\def\fewcomponents{s}
\def\lastc{\alpha} 
\def\sprob{\xi}  
\def\tol{\eta}
\def\es{\edgesubset}
\def\esl{\underline\es}
\def\esu{\overline\es}

\def\largec{\mathcal{L}}
\def\weighty{\mathcal{W}}

\def\Chat{\widehat C}

\def\boldgamma{\boldsymbol{\gamma}}  
 
\def\citeyear{\cite}
\def\citeNP{\cite} 

\newtheorem{theorem}{Theorem}
\newtheorem{lemma}[theorem]{Lemma}
\newtheorem{corollary}[theorem]{Corollary}
\newtheorem{observation}[theorem]{Observation}

\title[Approximating the Potts partition function]{Approximating the partition function\\
of the ferromagnetic Potts model}
\thanks{This work was partially supported by the EPSRC grant 
\it The Complexity of Counting in Constraint Satisfaction Problems}

\author{Leslie Ann Goldberg}
\address{Leslie Ann Goldberg, Department of Computer Science,
University of Liverpool, Ashton Building,
Liverpool L69 3BX, United Kingdom.}

\author{Mark Jerrum}
\address{Mark Jerrum, School of Mathematical Sciences\\
Queen Mary, University of London, Mile End Road, London E1 4NS, United Kingdom.}
         
\begin{document}

\begin{abstract}
We provide evidence that it is computationally difficult 
to approximate the partition function of the
ferromagnetic $q$-state Potts model when $q>2$.  Specifically
we show that the partition function is hard for the complexity class
$\RHPi$ under approximation-preserving reducibility.
Thus, it is as hard to approximate the partition function 
as it is to find approximate
solutions to a wide range of counting problems, including
that of determining the number of independent sets in a bipartite graph.
Our proof
exploits the 
first order
phase transition 
of the ``random cluster'' model, which is a probability distribution 
on graphs that is closely related to the $q$-state Potts
model. 
\end{abstract}
  
  \maketitle
 
\section{Introduction}
Let $q$ be a positive integer. 
The $q$-state Potts partition function of a graph~$G=(V,E)$,
with uniform interactions of strength~$\gamma\geq-1$ along the edges, is
defined as
\begin{equation}\label{eq:PottsGph}
\ZPotts(G;q,\gamma) = 
\sum_{\sigma:V\rightarrow [q]}
\prod_{e=\{u,v\}\in E}
\big(1+\gamma\,\delta(\sigma(u) ,\sigma(v))\big),
\end{equation}
where $[q]=\{1,\ldots,q\}$ is a set of $q$~spins or colours,
and $\delta(s,s')$ is~$1$ if $s=s'$, and 0 otherwise.
The partition function is a sum over ``configurations''~$\sigma$
which assign spins to vertices in all possible ways.
Later, we shall widen this definition to allow the interaction strength $\gamma$
to be a function of the edge~$e$, and allow $G$ to be a hypergraph (a generalisation
of graph in which edges contain an arbitrary number of vertices), but the 
above restricted case is sufficient for this overview.  Mostly we shall concentrate in this
paper on the ferromagnetic situation, characterised by $\gamma>0$. In the 
ferromagnetic Potts model, configurations $\sigma$ with many adjacent like spins
make a greater contribution to the partition function $\ZPotts(G;q,\gamma)$ than
those with few.

The statistical mechanical model just described was introduced 
by Potts~\citeyear{Potts} and generalises the classical Ising model from 
two to $q$~spins.  

Definition (\ref{eq:PottsGph}) applies only when $q$ is a positive integer.
However, it transpires that, regarding $q$ as an indeterminate, (\ref{eq:PottsGph})
defines a polynomial in~$q$, and in this way we can make sense of the Potts partition
function for non-integer~$q$, even though the underlying physical model has 
no meaning.
An equivalent, but more concrete way of approaching the partition function
when $q$ is non-integer
is via the Tutte polynomial, which in its ``random cluster'' formulation is 
defined as follows:
\begin{equation}\label{eq:TutteGph}
\ZTutte(G;q,\gamma)=\sum_{F\subseteq E}
q^{\kappa(V,F)}\gamma^{|F|},
\end{equation}
where $\kappa(V,F)$ denotes the number of connected components in
the graph $(V,F)$.  
The notation is as before, except that now $q$ is an arbitrary real number.
Again, for simplicity, we are assuming that $G$ is a graph (not
a hypergraph) and that the edge weight $\gamma$ is uniform over edges.
The multivariate definition can be guessed at and will in any case appear 
later in the paper.  For readers who are familiar with the classical 
$(x,y)$-parameterisation
of the Tutte polynomial 
(see [\citeNP{Tutte84}; \citeNP{welsh}]),  the transformation 
between that and the one here is given by $\gamma=y-1$ and $q=(x-1)(y-1)$.

Although (\ref{eq:PottsGph}) and (\ref{eq:TutteGph}) are formally very
different, they define the same polynomial in~$q$: 
see Observation~\ref{obs:FK}.  We continue the discussion now in terms 
of the Tutte polynomial~(\ref{eq:TutteGph}), remembering all along that 
we include as a special case the Potts partition function, and as an even 
more special case that of the Ising model.  We denote by $\Tutte(q,\gamma)$
the computational task of computing $\ZTutte(G;q,\gamma)$ given a graph $G$ as
problem instance.  Then each pair $(q,\gamma)$ defines a separate computational
problem, and we can study the computational complexity of this problem as $q$ and~$\gamma$
vary.  It is important to note that $q$ and $\gamma$ do not form part of the problem instance,
which consists simply of the graph~$G$.
For the purposes of this discussion, we may assume that $q$ and $\gamma$ are rational,
in order to avoid representation issues, but in the main body of the paper we 
work in the wider class of ``efficiently approximable'' real numbers.

In a seminal paper, Jaeger, Vertigan and Welsh~\citeyear{JVW90} examined
the problem of computing $\ZTutte(G;q,\gamma)$ exactly.  In the exact setting,
they completely classified the complexity of $\Tutte(q,\gamma)$ 
for all $q,\gamma$ (in fact for all complex $q,\gamma$).
It transpires that $\Tutte(q,\gamma)$ is $\numP$-hard (i.e., as hard as determining the 
number of satisfying assignments to a CNF Boolean formula), except when $q=1$,
or when $(q,\gamma)$ is one of a finite number of ``special points'';  in these cases 
$\Tutte(q,\gamma)$ is polynomial-time computable.

In light of Jaeger et al.'s strong negative result, attention turned to the question
of whether $\ZTutte(G;q,\gamma)$ could be approximated with arbitrarily small
specified relative error.  In the context of computing partition functions,
the appropriate notion of efficient approximate computation is the ``Fully polynomial
randomised approximation scheme''  or FPRAS, 
which is rigorously defined in~\S\ref{sec:FPRAS}.  An early positive 
result was provided by Jerrum and Sinclair~\citeyear{JS93}, who presented an FPRAS 
for the case $q=2$ and $\gamma>0$, that is to say, for the ferromagnetic Ising model.  
Sadly, no further generally applicable positive results have appeared since then, 
though FPRAS's have been proposed for restricted classes of graphs, e.g., dense 
or degree-bounded \cite{dense}.

Greater progress has been made in the negative direction.  
Goldberg and Jerrum~\citeyear{tuttepaper}
showed, under the reasonable complexity-theoretic assumption $\mathrm{RP}\not=\mathrm{NP}$,
that no FPRAS exists for $\Tutte(q,\gamma)$ for a wide range of values of the parameters 
$(q,\gamma)$.  Stated informally, $\mathrm{RP}\not=\mathrm{NP}$ is the assumption that 
there are problems in NP that cannot be decided by a polynomial-time randomised 
algorithm.  Intuitively, it is only slightly stronger than the more usual 
$\mathrm{P}\not=\mathrm{NP}$ assumption.  
Note that the claim that no FPRAS exists 
is stronger than it may appear at first sight.  By a folklore observation (see \S\ref{sec:FPRAS}),
ruling out an FRPAS also rules out polynomial-time approximation algorithms with constant
(or even polynomial) factor error.   
As an indicative example of what is known about the Tutte polynomial, 
the intractability result of \cite{tuttepaper} covers the entire half-plane 
$\gamma<-2$  
except for the tractable case $q=1$ and 
the case $q=2$ where the problem is equivalent to approximately
counting perfect matchings. Similar results apply when $q/\gamma<-2$.
The restriction to planar graphs was treated in a follow-up
paper \cite{planartutte}.  However none of the existing intractability results apply 
to the region $q>0$ and $\gamma>0$ that concerns us here, and which
is perhaps the one of greatest physical interest.

Our goal here is to present the first evidence that $\Tutte(q,\gamma)$
is computationally hard in the region $q>2$ and $\gamma>0$, i.e., the 
region corresponding to the ferromagnetic Potts model with $q>2$ states.
We achieve this, but under a stronger complexity-theoretic assumption than
$\mathrm{RP}\not=\mathrm{NP}$.  To explain this assumption, a digression into
computational complexity is required.  

The complexity class $\RHPi$ 
of counting problems was introduced by 
Dyer,  Goldberg, Greenhill and Jerrum~\citeyear{APred} as a means  
of classifying a wide class of approximate counting problems that were
previously of indeterminate computational complexity.  The problems in $\RHPi$
are those that 
can be expressed in terms of counting the number of models of a logical formula
from a certain syntactically restricted class.  
(Although the authors were not aware of it at the time, this 
syntactically restricted 
class had already been studied under the title ``restricted Krom SNP'' \cite{Dalmau05}.
Yet another terminological variation  
is to say that problems in $\RHPi$  
enumerate solutions to a linear Datalog program [\citeNP{Dalmau05}; \citeNP{Datalog}].)
The complexity class $\RHPi$ has a completeness class (with respect to 
approximation-preserving ``AP-reductions'') 
which includes a wide and ever-increasing
range of natural counting problems, including:  
independent sets in a bipartite graph,
downsets in a partial order, 
configurations in the Widom-Rowlinson model (all \cite{APred}),
the partition function of the 
ferromagnetic Ising model with mixed external field (i.e., not
consistently favouring one or other spin) \cite{ising},
and stable matchings \cite{stablematchings}.
Either all of these problems admit an FPRAS (i.e., are efficiently approximable),
or none do.  No FPRAS is known for any of them at the time of writing, despite 
much effort having been expended on finding one.  

All the problems in the completeness class mentioned above are inter-reducible
via AP-reductions, so 
any
of them could be said to exemplify the completeness class.
However, mainly for historical reasons, 
the particular problem $\BIS$, of counting independent sets
in a bipartite graph, tends to be taken as the exemplar of the class, 
much in the same way that $\Sat$ has
a privileged status in the theory of NP-completeness.
Our main result is:
\begin{theorem}
\label{thm:main}
Suppose that
$q>2$ and $\gamma>0$ are efficiently approximable. Then
$\BIS\APred \Tutte(q,\gamma)$.
\end{theorem}
Here, $\APred$ is the symbol for ``is AP-reducible to'', and ``efficiently
approximable'' is a concept defined in \S\ref{sec:FPRAS};  
suffice it to say for now that 
the rational numbers are trivially efficiently approximable.

One limitation of our result is that the 
resulting inapproximability of the problem $\Tutte(q,\gamma)$
 is conditional on there being no FPRAS for $\BIS$ (and the rest of the completeness
class), rather than on the weaker assumption
$\mathrm{NP}\not=\mathrm{RP}$.
In fact, we conjecture that $\BIS$ does not admit an 
FPRAS. The basis for our conjecture is empirical ---
namely 
that the collection of known $\BIS$-equivalent problems is growing
and that the problem itself has survived its first decade 
despite considerable efforts to find an FPRAS.
For example, 
Ge and \v Stefankovi\v c~\citeyear{BISpoly}
recently proposed an interesting new MCMC algorithm for sampling indepedent sets
in bipartite graphs. Unfortunately, however, the relevant Markov chain mixes slowly \cite{BISpolyslow}
so even this interesting new idea does not give an FPRAS.

Despite the fact that our results are limited by a strong complexity-theoretic assumption,
we feel there are counterbalancing strengths that 
justify this investigation.  One is the range and intrinsic interest of
the problem under consideration.  Whether in the guise of the Potts partition 
function, or
of the Tutte plane, the computational complexity of $\Tutte(q,\gamma)$
has received considerable attention since it was first studied by Jaeger
et al.~\citeyear{JVW90}:  see, for example, 
[\citeNP{dense}; \citeNP{bipartite}; \citeNP{planar}; \citeNP{welsh}]. 
So it seems worth striving for a complexity classification
even under 
a strong assumption such
as the one we are making.  
The situation is similar to working
with the Unique Games Conjecture in the area of approximation algorithms
for optimisation problems, or employing the class PPAD in analysing
the complexity of Nash equilibria.
Futhermore, Theorem~\ref{thm:main}
has a wide range of applicability, 
covering as it does the whole region $q>2$, $\gamma>0$, which, in the classical 
parameterisation of the Tutte polynomial, equates to the entire upper quadrant
of the Tutte plane above the positive branch of the hyperbola 
$H_2=\{(x,y):(x-1)(y-1)=2\}$.  Note that the 
$\BIS$-hard
region extends right to the tractable hyperbola~$H_2$.

Another potential strength of the work
is that the reduction introduces a novel technique
that may have wider applicability.  The idea is conceptually simple and can
be sketched informally here.  In the first step, we reduce $\BIS$ to a 
hypergraph version of the Tutte polynomial.  (The conventional Tutte polynomial may be 
recovered as the specialisation to 2-uniform hypergraphs.)  This step, if 
not routine, is at least standard in its techniques.  After this, we
show how to
simulate each hyperedge containing $t$~vertices by a graph gadget 
with $t$ distinguished vertices or terminals.  

At this point we exploit the 
first order phase transition that is a feature of the so-called 
random cluster model when $q>2$.  The configurations
of the random cluster model on a graph~$G$ are spanning subgraphs of~$G$, 
which are weighted
according to the numbers of edges and connected components they contain.
As formulated in (\ref{eq:TutteGph}), the Tutte polynomial
is the partition function of this model.  
The gadget is designed and carefully tuned so that is has 
two coexisting ``phases'':  one in which the random cluster 
configurations (spanning subgraphs) have a large connected (or 
``giant'') component, and one in which they don't.  We show that it is 
possible to arrange for the $t$~terminals to be (with high probability)
in a single component
in one phase and in $t$~distinct components in the other.  
This provides us with a bistable gadget that simulates a potentially
large hyperedge using many 2-vertex edges.
Note that AP-reductions often exploit phase transitions, playing 
one class of configurations off against
another.  See the examples in \cite{APred} and \cite{KelkPhD}.
What is new here is using the complex phase transitions that arise in
actual models from statistical physics to derive new complexity
results.\footnote{Shortly after the preliminary version of this paper \cite{confversion}, Sly \citeyear{Sly} independently 
obtained strong hardness results for the hard-core model using a
gadget that had previously been used only to rule out fast MCMC
algorithms.}
Unfortunately, the delicate
nature of the gadgets needed to  exploit this kind of phase transition does
lead to significant technical complexity in our analysis.

Non-trivial phase transitions have been used 
in the past to rule out certain natural Markov chain Monte Carlo (MCMC) 
approaches to approximate counting problems.  Indeed, the 
first-order
phase 
transition in the random cluster model was already exploited 
by Gore and Jerrum~\citeyear{GJ99} to demonstrate that the Swendsen-Wang algorithm 
is not always effective for the ferromagnetic Potts model.  
(See also Borgs et al.~\citeyear{torpid} for a more thorough working out 
of this idea.)
In those applications
the aim was to rule out a certain algorithmic approach, namely MCMC, whereas here
our goal is to show inherent intractability.

Finally, note that Theorem~\ref{thm:main} establishes $\BIS$-hardness of $\Tutte(q,\gamma)$ but not $\BIS$-equivalence.
It would be very interesting to know whether there is an AP-reduction from $\Tutte(q,\gamma)$ to $\BIS$.
Note that the complexity of approximate counting is complicated. Bordewich~\citeyear{bordewich} has shown that 
if any problem 
in~\#P
fails to have an FPRAS, then there is an infinite approximation hierarchy within~\#P.

The rest of the paper is laid out as follows.  Sections~\ref{sec:TuttePotts}
and~\ref{sec:rc} introduce the models we shall be working with.
Section~\ref{sec:natural} is the technical heart of the paper.
It analyses the random cluster model on some 
families of graphs, building up to the constuction of the hyperedge simulation
gadget.  Sections \ref{sec:FPRAS} and~\ref{sec:APred} 
cover the computational framework 
we work in, while \S\ref{sec:results}--\ref{shift} present
and analyse the reductions themselves.  
The final \S\ref{sec:3uniform} turns the hypergraph intractability
result on its head by focusing on the $3$-uniform case. By reduction
to the graphic case, we obtain an FPRAS for the partition function.

\section{The Tutte polynomial of a hypergraph and the Potts partition function}
\label{sec:TuttePotts}
Let $\hypergraph=(\hypervertices,\hyperedges)$ be a hypergraph
with vertex set $\hypervertices$ and hyperedge (multi)set~$\hyperedges$.
Following the usual convention 
for the Tutte polynomial \cite{Sokal05},
a hypergraph is allowed to have parallel edges. This is why $\hyperedges$
is a multiset. 
The multivariate Tutte polynomial of $\hypergraph$ is 
defined as follows
$$
\ZTutte(\hypergraph;q,\boldgamma)=\sum_{\calF\subseteq\hyperedges}
q^{\kappa(\hypervertices,\calF)}
\prod_{\hyperedge\in\calF}
\gamma_\hyperedge,
$$
where 
$\boldgamma=\{\gamma_\hyperedge\}_{\hyperedge\in\hyperedges}$ 
and
$\kappa(\hypervertices,\calF)$ denotes the number 
of connected components in the subhypergraph $(\hypervertices,\calF)$.
(Two vertices $u$, $v$, are in the same component of $(V,\calF)$ if $u=v$, or there
is a sequence $\hyperedge_1,\ldots,\hyperedge_\ell\in\calF$ of hyperedges with 
$u\in \hyperedge_1$, $v\in \hyperedge_\ell$ and $\hyperedge_i\cap f_{i+1}\not=\emptyset$ for
$1\leq i<\ell$.)
This partition function was studied (under a different name) by Grimmett~\citeyear{Grimmett}.
An undirected graph $\graph$ can be viewed as a $2$-uniform hypergraph 
(a hypergraph in 
which every hyperedge has size~$2$).
In this case, $\ZTutte(\graph;q,\boldgamma)$ coincides with the usual definition of the multivariate Tutte polynomial~\cite{Sokal05}.

Let $q$ be a positive integer. The $q$-state Potts partition function of $\hypergraph$ is
defined as follows:
$$
\ZPotts(\hypergraph;q,\boldgamma) = 
\sum_{\sigma:\hypervertices\rightarrow [q]}
\prod_{\hyperedge\in\hyperedges}
\big(1+\gamma_\hyperedge\delta(\{\sigma(\hypervertex) \mid \hypervertex\in \hyperedge\})\big),$$
where $[q]=\{1,\ldots,q\}$ is a set of $q$~spins or colours,
and $\delta(S)$ is~$1$ if its argument is a singleton and 0 otherwise.
The partition function is a sum ranging over 
assignments of spins to vertices, which are often referred to as ``configurations''.
 
The following observation is due to Fortuin and Kastelyn. 
\begin{observation}
\label{obs:FK}
If $q$ is a positive integer then
$\ZPotts(\hypergraph;q,\boldgamma) = \ZTutte(\hypergraph;q,\boldgamma)$.
\end{observation}
\begin{proof}
A proof can
be found in \cite[Proposition 3.1]{CSS} or in \cite{Grimmett}, but we include it here
for completeness.
The argument is a straightforward generalisation 
of the standard proof for the graph case. See \cite[Theorem 2.3]{Sokal05}.
Consider 
$$\sum_{\sigma: \hypervertices\rightarrow [q]}
\sum_{\calF\subseteq\hyperedges}
\prod_{\hyperedge\in \calF} 
\big(\gamma_\hyperedge\delta(\{\sigma(\hypervertex) \mid \hypervertex\in \hyperedge)\big).$$
Now ``integrate out'' (i.e., explicitly sum over) $\calF$ to see that this quantity is equal to
$\ZPotts(\hypergraph;q,\boldgamma)$ and
integrate out $\sigma$ to see that this quantity is equal to 
$\ZTutte(\hypergraph;q,\boldgamma)$.
\end{proof}

The Potts model is said to be 
\emph{ferromagnetic} if   the edge weights $\gamma_\hyperedge$ are non-negative. In this case, a monochromatic
edge contributes more weight than an edge with multiple spins.

For a subset $\calF\subseteq \hyperedges$ of the hyperedges of a graph,
we use $\gamma({\calF})$ to denote $\prod_{\hyperedge\in\calF}
\gamma_\hyperedge$.

\section{The Random Cluster Model}
\label{sec:rc}
  
Consider a graph $\graph=(\graphvertices,\graphedges)$.   
Every edge $e\in\graphedges$ is associated with a
quantity $\edgeprob(e)\in[0,1]$.
Then for a set of edges $\edgesubset\subseteq\graphedges$ define 
$$
\tildeP(\graph;\edgesubset,q,\edgeprob) = q^{\kappa(\graphvertices,\edgesubset)}
\prod_{e\in \edgesubset} \edgeprob(e) 
\prod_{e \in \graphedges\setminus \edgesubset} (1-\edgeprob(e)).
$$
Let 
$$
\Zrc(\graph;q,\edgeprob) = \sum_{\edgesubset\subseteq \graphedges} \tildeP(\graph;\edgesubset,q,\edgeprob)
= \ZTutte(G;q,\boldgamma)\prod_{e\in \graphedges}(1-\edgeprob(e)),
$$
where $\gamma_e = \edgeprob(e)/(1-\edgeprob(e))$.
Then the probability of edge-set~$\edgesubset$ in the random cluster model is given by
$$
P(\graph;\edgesubset,q,\edgeprob) = \frac{\tildeP(\graph;\edgesubset,q,\edgeprob)}
{\Zrc(\graph;q,\edgeprob)}.
$$

The \emph{random cluster model} refers to the distribution
$\RC(\graph;q,\edgeprob)$, in which
a subset $\edgesubset$ of edges is chosen
 with probability 
$P(\graph;\edgesubset,q,\edgeprob)$.
The difference between $\ZTutte$ and $\Zrc$ is simply one of parameterisation.
Nevertheless, the change of parameter is useful, as it allows us to employ
probabilistic terminology and exploit existing results from the random graph
literature.

A central ingredient which helps us to understand the random cluster model is to
compare it to the (multivariate) Erd\H os-R\'enyi model of a random graph.
Consider a graph $\graph=(\graphvertices,\graphedges)$.   
Every edge $e\in\graphedges$ is associated with a
quantity $\edgeprob'(e)\in[0,1]$.
In 
the Erd\H os-R\'enyi model $\ER(\graph;\edgeprob')$, a subset $\edgesubset$ 
of edges
is chosen with probability 
$$\prod_{e \in \edgesubset} \edgeprob'(e)
\prod_{e \in \graphedges \setminus\edgesubset} (1-\edgeprob'(e)).$$
Thus, to choose a configuration $\edgesubset$ in the $\ER(G;\edgeprob')$ model, each edge $e\in \graphedges$ is included 
in~$\edgesubset$ independently with probability $\edgeprob'(e)$.

\subsection{Stochastic domination results}

We will use several  simple 
stochastic domination results,
which are close in spirit, and indeed in their proofs, to results of 
Holley~\citeyear{Holley};  see, in particular, Theorem~(6) of that article
and its proof.

If $\edgesubset^+$ and $\edgesubset^-$ are disjoint subsets of $\graphedges$,
let $\RC(\graph;q,\edgeprob;\edgesubset^+,\edgesubset^-)$ 
be the random cluster model
conditioned on the fact that the chosen subset $\edgesubset$ 
contains every edge in $\edgesubset^+$ and no edges in 
$\edgesubset^-$.
(To avoid trivialities, we assume that no edge~$e\in\edgesubset^+$
has $\edgeprob(e)=0$ and that no edge $e\in \edgesubset^-$ has $\edgeprob(e)=1$.)
Similarly, let 
$\ER(\graph;\edgeprob';\edgesubset^+,\edgesubset^-)$
be the 
Erd\H os-R\'enyi model
with this conditioning.
We  give conditioned versions of the stochastic domination results.

\begin{lemma}
Consider a graph $\graph=(\graphvertices,\graphedges)$
in which each edge $e\in\graphedges$ is associated with a
quantity $\edgeprob(e)\in[0,1]$.
Let $\edgesubset^+$ and $\edgesubset^-$ be disjoint subsets of $\graphedges$
such that $\edgeprob(e)>0$ for $e\in\edgesubset^+$ 
and $\edgeprob(e)<1$ for $e\in\edgesubset^-$.
Suppose $q\geq 1$. Then
the random cluster model 
$\RC(\graph;q,\edgeprob;\edgesubset^+,\edgesubset^-)$ 
is stochastically dominated by the Erd\H os-R\'enyi model
$\ER(\graph;\edgeprob;\edgesubset^+,\edgesubset^-)$
in the sense that we can select a pair
$(\edgesubset,\edgesubset')$ 
such that $\edgesubset$ is drawn from 
$\RC(\graph;q,\edgeprob;\edgesubset^+,\edgesubset^-)$ 
and $\edgesubset'$ is drawn from 
$\ER(\graph;\edgeprob;\edgesubset^+,\edgesubset^-)$
and $\edgesubset \subseteq \edgesubset'$.
\label{lem:coupleone}
\end{lemma}

\begin{proof}
Let $\mathcal{M}$ be the ``heat-bath on an edge'' Markov chain
for $\RC(\graph;q,\edgeprob;A^+,A^-)$.
The states of this chain are subsets of edges. From state~$A$, the chain chooses an edge $e$~u.a.r.
It then moves to a state in $\{A\cup\{e\},A\setminus e\}$. (Note that one of these states
 is identical to~$A$.)
The probability of moving to each state~$A'$ is proportional to 
the probability of~$A'$ in the distribution~$\RC(\graph;q,\edgeprob;A^+,A^-)$.
Let 
$\mathcal{M}'$ be the ``heat-bath on an edge'' Markov chain
for $\ER(\graph;\edgeprob;A^+,A^-)$,
which is defined similarly.
Start $\mathcal{M}$ in state $\edgesubset_0$
and $\mathcal{M}'$ in state $\edgesubset'_0$
where $\edgesubset_0 \subseteq \edgesubset'_0$
and $\edgesubset_0$ and $\edgesubset'_0$ contain every edge in $\edgesubset^+$
and no edges in $\edgesubset^-$. We can then couple the evolution 
of the chains from the $i$'th pair of
states $(\edgesubset_i,\edgesubset'_i)$ as follows, guaranteeing
that $\edgesubset_{i+1} \subseteq \edgesubset'_{i+1}$. Choose the
same edge~$e\in\graphedges \setminus \edgesubset^+\cup \edgesubset^-$ in both chains.
In $\mathcal{M}$ the probability of 
putting $e$ in $\edgesubset_{i+1}$ is
either 
$\edgeprob(e)$
or 
$\edgeprob(e)/(\edgeprob(e)+ q(1-\edgeprob(e)))$,
each of which is
at most 
$\edgeprob(e)$,
which is the probability of putting it in 
$\edgesubset_{i+1}'$.
\end{proof}

\begin{lemma} 
Consider a graph $\graph=(\graphvertices,\graphedges)$
in which each edge $e\in\graphedges$ is associated with a
quantity $\edgeprob(e)\in[0,1]$.
For each edge $e \in\graphedges$, let $\edgeprob'(e)=\edgeprob(e)/q$.
Let $\edgesubset^+$ and $\edgesubset^-$ be disjoint subsets of $\graphedges$
such that $\edgeprob(e)>0$ for $e\in\edgesubset^+$ 
and $\edgeprob(e)<1$ for $e\in\edgesubset^-$.
Suppose $q\geq 1$. Then
the Erd\H os-R\'enyi model
$\ER(\graph;\edgeprob';\edgesubset^+,\edgesubset^-)$
is stochastically dominated by
the random cluster model 
$\RC(\graph;q,\edgeprob;\edgesubset^+,\edgesubset^-)$ 
in the sense that we can select a pair
$(\edgesubset,\edgesubset')$ 
such that $\edgesubset$ is drawn from 
$\ER(\graph;\edgeprob';\edgesubset^+,\edgesubset^-)$
and $\edgesubset'$ is drawn from  
$\RC(\graph;q,\edgeprob;\edgesubset^+,\edgesubset^-)$ 
and $\edgesubset \subseteq \edgesubset'$.
\label{lem:coupletwo}
\end{lemma}

\begin{proof}
Let $\mathcal{M}$ be the ``heat-bath on an edge'' Markov chain
for 
$\ER(\graph;\edgeprob';\edgesubset^+,\edgesubset^-)$
and let 
$\mathcal{M}'$ be the ``heat-bath on an edge'' Markov chain
for 
$\RC(\graph;q,\edgeprob;\edgesubset^+,\edgesubset^-)$ 
Start $\mathcal{M}$ in state $\edgesubset_0$
and $\mathcal{M}'$ in state $\edgesubset'_0$
where $\edgesubset_0 \subseteq \edgesubset'_0$
and $\edgesubset_0$ and $\edgesubset'_0$ contain every edge in $\edgesubset^+$
and no edges in $\edgesubset^-$. We can then couple the evolution 
of the chains from the $i$'th pair of
states $(\edgesubset_i,\edgesubset'_i)$ as follows, guaranteeing
that $\edgesubset_{i+1} \subseteq \edgesubset'_{i+1}$. Choose the
same edge~$e\in\graphedges \setminus \edgesubset^+\cup \edgesubset^-$ in both chains.
In $\mathcal{M}'$ the probability of 
putting $e$ in $\edgesubset'_{i+1}$ is
either 
$\edgeprob(e)$
or 
$\edgeprob(e)/(\edgeprob(e)+ q(1-\edgeprob(e)))$,
each of which is
at least 
$\edgeprob'(e)$,
which is the probability of putting it in 
$\edgesubset_{i+1}$.
\end{proof}

\begin{lemma}
\label{lem:coupleRC}
Consider a graph $\graph=(\graphvertices,\graphedges)$
in which each edge $e\in\graphedges$ is associated with a
quantity $\edgeprob(e)\in[0,1]$.
Suppose that, for each edge $e \in\graphedges$, $\edgeprob'(e)\geq\edgeprob(e)$.
Let $\edgesubset^+$ and $\edgesubset^-$ be disjoint subsets of $\graphedges$
such that $\edgeprob(e)>0$ for $e\in\edgesubset^+$ 
and $\edgeprob(e)<1$ for $e\in\edgesubset^-$.
Suppose $q\geq 1$. Then
 the random cluster model
$\RC(\graph;q,\edgeprob;\edgesubset^+,\edgesubset^-)$
is stochastically dominated by
the random cluster model 
$\RC(\graph;q,\edgeprob';\edgesubset^+,\edgesubset^-)$ 
in the sense that we can select a pair
$(\edgesubset,\edgesubset')$ 
such that $\edgesubset$ is drawn from 
the distribution
$\RC(\graph;q,\edgeprob;\edgesubset^+,\edgesubset^-)$
and $\edgesubset'$ is drawn from  
$\RC(\graph;q,\edgeprob';\edgesubset^+,\edgesubset^-)$ 
and $\edgesubset \subseteq \edgesubset'$.
\end{lemma}

\begin{proof}
Let $\mathcal{M}$ be the ``heat-bath on an edge'' Markov chain
for $\RC(\graph;q,\edgeprob;A^+,A^-)$
and let 
$\mathcal{M}'$ be the ``heat-bath on an edge'' Markov chain
for $\RC(\graph;q,\edgeprob';\edgesubset^+,\edgesubset^-)$ 
Start $\mathcal{M}$ in state $\edgesubset_0$
and $\mathcal{M}'$ in state $\edgesubset'_0$
where $\edgesubset_0 \subseteq \edgesubset'_0$
and $\edgesubset_0$ and $\edgesubset'_0$ contain every edge in $\edgesubset^+$
and no edges in $\edgesubset^-$. We can then couple the evolution 
of the chains from the $i$'th pair of
states $(\edgesubset_i,\edgesubset'_i)$ as follows, guaranteeing
that $\edgesubset_{i+1} \subseteq \edgesubset'_{i+1}$. Choose the
same edge~$e\in\graphedges \setminus \edgesubset^+\cup \edgesubset^-$ in both chains.
If the endpoints of the edge are in the same component of $\edgesubset_i-e$
then they are in the same component of $\edgesubset'_i-e$.
The probability of putting $e$ in $\edgesubset_{i+1}$ is
$\edgeprob(e)$
and the probability of putting $e$ in $\edgesubset'_{i+1}$ is
$\edgeprob'(e)$,
which is at least as big.
Otherwise, the probability of putting $e$ in $\edgesubset_{i+1}$ is
$\edgeprob(e)/(\edgeprob(e)+ q(1-\edgeprob(e)))$.
This is at most 
$\edgeprob'(e)/(\edgeprob'(e)+ q(1-\edgeprob'(e)))$
and the probability of putting $e$ in
$\edgesubset'_{i+1}$ is at least 
$\edgeprob'(e)/(\edgeprob'(e)+ q(1-\edgeprob'(e)))$.
\end{proof}

\subsection{The Fundamental Lemma of Bollob\'as, Grimmett and Janson}

Another way to compare the random cluster model to the 
Erd\H os-R\'enyi model is to use the following lemma, which is
a multivariate version of Bollob\'as, Grimmett and Janson's
``Fundamental Lemma'', \cite[Lemma~3.1]{BGJ}. 
To do the comparison, the lemma should be applied with $r=1/q$
so that the distribution $\RC(\induced{\graph}{V_1} ; r q,\edgeprob)$ in the statement of the lemma below
is identical to the distribution $\ER(\induced{\graph}{V_1} ;\edgeprob)$.

Let $0\leq r \leq 1$ be fixed. 
Given a subset $\edgesubset$ of edges chosen from
$\RC(\graph;q,\edgeprob)$, colour each component of
$(\graphvertices,\edgesubset)$ either \emph{red}, with probability~$r$,
or \emph{green}, with probability~$1-r$: different components are coloured independently of one another. The union of the red components is the \emph{red subgraph}
and the union of the green components is the \emph{green subgraph}.
Let $R$ be the set of vertices in the red subgraph.
Given a subset $U$ of vertices, let $\induced{\graph}{U}$
be the subgraph of~$\graph$ induced by~$U$.

\begin{lemma}
\label{lem:FL} (Bollob\'as, Grimmett and Janson \citeyear[Lemma~3.1]{BGJ}.)
Let $\graphvertices_1 \subseteq \graphvertices$.
Conditioned on $R=\graphvertices_1$, the red subgraph is distributed
according to $\RC(\induced{\graph}{V_1}; r q,\edgeprob)$
and the green subgraph is distributed according
to $\RC(\induced{\graph}{\graphvertices\setminus {V_1}};(1-r)q,\edgeprob)$. Conditioned on $R=\graphvertices_1$, the red subgraph and green subgraph are independent of each other.
\end{lemma}

\begin{proof}
Let $\graphvertices_2 = \graphvertices\setminus \graphvertices_1$.
Let $\graphedges_1 = \graphedges \cap \graphvertices_1^{(2)}$
and $\graphedges_2 = \graphedges \cap \graphvertices_2^{(2)}$,
where the notation $V_1^{(2)}$ denotes the set of unordered pairs 
of vertices from~$V_1$.  
Now consider $\edgesubset_1 \subseteq \graphedges_1$
and $\edgesubset_2 \subseteq \graphedges_2$.
The (unconditional)
probability that the red subgraph is $(\graphvertices_1,\edgesubset_1)$
and the green subgraph is $(\graphvertices_2,\edgesubset_2)$
is 
$$
\frac{\tildeP(\graph;\edgesubset_1\cup\edgesubset_2,q,\edgeprob)\,
r^{\kappa(\graphvertices_1,\edgesubset_1)}
{(1-r)}^{\kappa(\graphvertices_2,\edgesubset_2)}}
{\Zrc(\graph;q,\edgeprob)}
$$
so the conditional probability, conditioned on $R=V_1$
is
\begin{equation}
\label{eq:ref}
\frac{\tildeP(\graph;\edgesubset_1\cup\edgesubset_2,q,\edgeprob)\,
r^{\kappa(\graphvertices_1,\edgesubset_1)}
{(1-r)}^{\kappa(\graphvertices_2,\edgesubset_2)}}
{
\sum_{\edgesubset'_1\subseteq \graphedges_1,\edgesubset'_2\subseteq \graphedges_2}
{\tildeP(\graph;\edgesubset'_1\cup\edgesubset'_2,q,\edgeprob)\,
r^{\kappa(\graphvertices_1,\edgesubset'_1)}
{(1-r)}^{\kappa(\graphvertices_2,\edgesubset'_2)}}
} \end{equation}
which is
$$ P(\induced{\graph}{V_1}; \edgesubset_1,r q,\edgeprob)
P(\induced{\graph}{\graphvertices\setminus {V_1}};\edgesubset_2,(1-r)q,\edgeprob).
$$
To see this, note that the numerator of (\ref{eq:ref}) is
$$ 
q^{\kappa( \graphvertices,\edgesubset_1\cup \edgesubset_2)}\,
r^{\kappa(\graphvertices_1,\edgesubset_1)}
{(1-r)}^{\kappa(\graphvertices_2,\edgesubset_2)}
\prod_{e\in \edgesubset_1\cup \edgesubset_2} \edgeprob(e) 
\prod_{e \in \graphedges\setminus (\edgesubset_1 \cup \edgesubset_2)} (1-\edgeprob(e))
.$$
Since $\kappa( \graphvertices,\edgesubset_1\cup \edgesubset_2) = 
\kappa(\graphvertices_1,\edgesubset_1) + \kappa(\graphvertices_2,\edgesubset_2)$,
this can be decomposed as the product of  the three terms

$$q^{\kappa( \graphvertices_1,\edgesubset_1)} \,r^{\kappa(\graphvertices_1,\edgesubset_1)}
\prod_{e\in \edgesubset_1} \edgeprob(e) 
\prod_{e \in \graphedges_1\setminus \edgesubset_1  } (1-\edgeprob(e)) =
\tildeP(G[V_1];A_1,r q,p),
$$
$$
q^{\kappa(\graphvertices_2,\edgesubset_2)}\,
{(1-r)}^{\kappa(\graphvertices_2,\edgesubset_2)} 
\prod_{e\in \edgesubset_2} \edgeprob(e) 
\prod_{e \in \graphedges_2\setminus   \edgesubset_2} (1-\edgeprob(e))=
\tildeP(G[V_2];A_2,(1-r) q,p),
$$
and
$$\prod_{e \in \graphedges\setminus  (\graphedges_1 \cup \graphedges_2)} (1-\edgeprob(e)).$$
The denominator can be decomposed similarly.
\end{proof}

\section{The random cluster model on some natural graphs}
\label{sec:natural}
In this section we consider the random cluster model on a clique, and also on
a pair of connected cliques. The latter is used as a gadget in our constructions.
First, however, we need a technical lemma.

\subsection{A technical lemma} 

The following lemma is not about the random cluster model, but we will use
it in our analysis of the random cluster model on a pair of connected cliques.
 
\begin{lemma}\label{lem:yellowblue}
Suppose we have a partition of the set $[\nu]$ into $s$~blocks 
of size at most $\nu_{\max}$.
Randomly colour a subset of the elements of $[\nu]$
yellow according to a Bernoulli process with success probability $\sprob$.
Independently, randomly colour a subset of the elements of $[\nu]$ blue
according to a Bernoulli process with success probability $\sprob$ (so an element
can be coloured yellow, blue, both or neither).
We say that a block of the partition is {\em bicoloured} if it
contains both yellow and blue elements.  Then
$$
\Pr(\textup{no block is bicoloured})\leq 
   [(1-\sprob)^{\nu/s}(2-(1-\sprob)^{\nu/s})]^s
$$
(which is increasing  in $s$ and decreasing  in $\nu$) and
$$
\Pr(\textup{some block is bicoloured})
\leq
\nu[1-(1-\sprob)^{\nu_{\max}}]^2
$$
(which is increasing in both $\nu$ and $\nu_{\max}$).
\end{lemma}
\begin{proof}
Let the block sizes be $\nu_1,\ldots,\nu_s$, so that 
$\max_j\nu_j \leq \nu_{\max}$.
Observe that 
$$
\Pr(\text{$j$th block contains a yellow element})=1-(1-\sprob)^{\nu_j}
=\sprob_j\text{ (say)},
$$
and that the same bound applies, of course, to blue elements.  Thus,
using the convention that the index $j$ always ranges over $1\leq j\leq s$,
$$\Pr(\text{some block is bicoloured})\leq
   \sum_j\sprob_j^2\leq\sum_j[1-(1-\sprob)^{\nu_{\max}}]^2,$$
which gives the second bound.  Clearly, this bound is monotonically
increasing in $\nu$ and~$\nu_{\max}$ as claimed.

Now for the first bound.  
\begin{align}
\Pr(\text{no block is bicoloured})&=\prod_j(1-\sprob_j^2)
=\prod_j(1-\sprob_j)\prod_j(1+\sprob_j)\notag\\
&=\prod_j(1-\sprob)^{\nu_j}\prod_j(1+\sprob_j)\notag\\
&=(1-\sprob)^\nu\prod_j(1+\sprob_j).\label{eq;opt2}
\end{align}
To get an upper bound we are interested in evaluating the supremum of 
\begin{equation}\label{eq:opt1}
\prod_j(1+\sprob_j)=\prod_j[2-(1-\sprob)^{\nu_j}]
\end{equation}
over the domain defined by the linear inequalities
$0\leq\nu_j\leq\nu$ and $\sum_j\nu_j=\nu$.
We consider this as an optimisation problem over~$\R^s$ even though 
the $\nu_j$ are all integers;  of course, this will if anything only
increase the supremum.  We are considering a continuous function 
over a closed, bounded set, so the supremum is achieved at some 
point;  we claim that this (unique) point is 
$\nu_1=\cdots=\nu_s=\nu/s$.  For if not, then at least one pair,
say $\nu_1$ and $\nu_2$ would be unequal.  But then it is easily checked
that replacing $\nu_1$ and $\nu_2$ by their average would increase 
$[2-(1-\sprob)^{\nu_1}][2-(1-\sprob)^{\nu_2}]$, and hence increase 
the right hand side of (\ref{eq:opt1}):  a contradiction.
Substituting $\nu_j=\nu/s$ into equations (\ref{eq:opt1})
and then~(\ref{eq;opt2}), we obtain
\begin{align}
\Pr(\text{no block is bicoloured})
&\leq (1-\sprob)^\nu[2-(1-\sprob)^{\nu/s}]^s\label{eq:ub1}\\
&=[(1-\sprob)^{\nu/s}(2-(1-\sprob)^{\nu/s})]^s,\label{eq:ub2}
\end{align}
as desired.

In only remains to verify the monotonicity claims about (\ref{eq:ub2}).
Let $u(\nu)=(1-\sprob)^{\nu/s}\in(0,1]$. Regarding $s>0$ as fixed, 
$u(\nu)$ decreases monotonically with $\nu$. 
Also, $[u(2-u)]^s$ increases monotonically as a function of $u$ in the 
range $(0,1)$. Thus, expression (\ref{eq:ub2}) decreases monotonically 
with $\nu$.

Now regard $\nu$ as fixed and make the change of variable $s=ax$,
where $a=\nu\ln((1-\sprob)^{-1})$, and note that $x>0$.
Then expression~(\ref{eq:ub1})
becomes $[e^{-1}(2-e^{-1/x})^x]^a$.  Thus, it is enough to show that
$f(x)=x\ln(2-e^{-1/x})$ increases monotonically with~$x$.
Now 
$$f''(x)=-2x^{-3}e^{-1/x}(2-e^{-1/x})^{-2}<0,$$
and
$$f'(x)=\ln(2-e^{-1/x})-x^{-1}e^{-1/x}(2-e^{-1/x})^{-1}\to0\quad
\text{as $x\to\infty$}.
$$
These two facts imply $f'(x)>0$ for $x>0$.
\end{proof}

\subsection{The random cluster model on a clique}
\label{sec:clique}
Bollob\'as, Grimmett and Janson~\citeyear{BGJ} studied the random cluster model on the complete $N$-vertex graph~$K_N$.  
More detailed analyses have since been performed, for example
by Luczak and \L uczak~\citeyear{LuczakLuczak}, but the  approach of the earlier paper is 
easier to adapt to our needs.

For fixed~$q$ and a fixed constant~$\lambda$, 
they studied the distribution $\RC(K_N,q,p)$ where $p$ is the constant function which
assigns every edge $e$ of $K_N$ the value $p(e)=\lambda/N$.
They show that there is a critical value~$\lambda_c$, depending on~$q$,
so that, if $\lambda>\lambda_c$ then, as $N\rightarrow \infty$, with high probability
a configuration~$A$ drawn from $\RC(K_N,q,p)$ 
will have a large component (of size linear in~$N$) and otherwise, with high probability 
the largest component will be much smaller (of size logarithmic in~$N$).

For $q>2$, the critical value $\lambda_c$ is defined as follows.
$$\lambda_c = 2\left(\frac{q-1}{q-2}\right)\ln(q-1).$$
It is important for our analysis that $\lambda_c < q$ (see \cite[p.16]{BGJ}, or 
by calculus).
We define $\delta = (q-\lambda_c)/2>0$
and $\lambda = \lambda_c + \delta$.
Let $\theta=(q-2)/(q-1)$.

We will use the following lemma, which follows from 
Theorem~2.2 and Equation~(5) of~\cite{BGJ}.
\begin{lemma}
\label{lem:clique}
Fix $q>2$ and define $\lambda$ and $\theta$ as above.  Let $p$ be the constant function which assigns
every edge $e$ of $K_N$ the value $p(e)=\lambda/N$.
Let $A$ be drawn from $\RC(K_N,q,p)$. The probability that $A$ has a connected component of size
at least $\theta N$  tends to~$1$ as $N\rightarrow \infty$.
\end{lemma}

\subsection{The random cluster model on a pair of connected cliques}
\label{sec:gadget}

Let $\gadget$ 
be the 
complete graph with vertex set $\gadgetvertices = \clique\cup \terminalset$.
Let $\gadgetedges$ denote the edge set of~$\gadget$ and let
$\cliquesize = |\clique|$ and $\terminals = |\terminalset|$.
Let $\clique^{(2)}$ denote the set of unordered pairs of distinct elements in~$\clique$  
and define $\terminalset^{(2)}$  similarly. 
Let $\critprob$ be a value in $[0,1]$.
Define $\edgeprob$ as follows:
$$\edgeprob(e) = \begin{cases}
\critprob, &\text{if $e\in \clique^{(2)}$,}\\
\cliquesize^{-3/4}, &\text{if $e\in\clique\times \terminalset$, and}\\
1, &\text{if $e\in \terminalset^{(2)}$.}
\end{cases}
$$
Here and in similar situations
in this paper, we slightly abuse notation by identifying $K \times T$ with the 
set of unordered pairs with one element from~$K$ and one from~$T$.

Ultimately, we will use the graph~$\gadget$ 
(or, more precisely, $\gadget$ with the edges $\terminalset^{(2)}$ deleted) as a 
gadget to simulate the contribution of a hyperedge on the set~$\terminalset$ to the multivariate
Tutte polynomial of~$\gadget$.
Thus, we refer to vertices in~$\terminalset$ as ``terminals'' of the graph~$\gadget$.
For a subset $\edgesubset\subseteq \gadgetedges$, let $\rv(\edgesubset)$ be the number of connected components in the graph 
$(\gadgetvertices,\edgesubset \setminus \terminalset^{(2)})$
that contain 
terminals.

A remark about the gadget~$\gadget$ and its eventual use.
When we come to use the gadget, the edges in $T^{(2)}$ will not be 
present.  It is for this reason that we are interested in the structure of 
connected components in $\gadget$ in the absence of these edges,
and specifically the random variable $\rv(\edgesubset)$.  However, it turns out that 
the key properties of the gadget are easier to verify if we work with 
a random cluster distribution associated with $\gadget$, exactly 
as given above, with the edges $T^{(2)}$ present.  Informally, the appropriate
``boundary condition'' for the gadget is the one in which the terminals are 
joined with probability~1.

The following two lemmas establish some useful properties of the gadget.
The second shows that the distribution of $\rv(\edgesubset)$ is concentrated at the extremes
of its range, i.e., $\rv=1$ or $\rv=\terminals$.  This concentration property holds for a wide 
range of values for the edge probability~$\critprob$.  The first lemma, which is easier, shows that we 
can tune~$\critprob$ so that the balance of probability between those extremes is
the desired value~$\gamma$.   Later, in Lemmas \ref{lem:computeZ} and~\ref{lem:computerho},
we shall show that a sufficiently 
close approximation to this~$\critprob$ can be efficiently computed.

\begin{lemma}\label{lem:excludedmiddleone}
Fix $q>2$ and  let $\lambda = \lambda_c + (q-\lambda_c)/2$.
Fix a weight $\gamma>0$ and let $\cliquesize_0$ be a sufficiently
large quantity depending on $q$ and $\gamma$.
Suppose a number of terminals $\terminals>1$ is given
and fix 
$\cliquesize\geq\max\{\terminals^{16}, \cliquesize_0\}$.
Then there is a parameter~$\critprob$
satisfying $\cliquesize^{-3} \leq \critprob \leq
\lambda/\cliquesize\leq \tfrac14$ 
such that, if $\edgesubset$ is drawn from $\RC(\gadget;q,p)$ then
\begin{equation}
\Pr(\rv(\edgesubset)=1)=\gamma\Pr(\rv(\edgesubset)=\terminals)\label{eq:balance}.
\end{equation}
\end{lemma}
\begin{proof}
Define   $\theta$ (depending on~$q$) as in Section~\ref{sec:clique}.
Let $\psi(\critprob) = \Pr(\rv(\edgesubset)=\terminals)/\Pr(\rv(\edgesubset)=1)$.
We will use stochastic domination to show
\begin{itemize}
\item $\psi(\critprob)$ is monotonically decreasing as a function of $\critprob$,
\item $\psi(\cliquesize^{-3})> 1/\gamma$, and
\item $\psi(\lambda/\cliquesize)< 1/\gamma$.
\end{itemize}
Since $\psi(\critprob)$ is a rational function in~$\critprob$ and the denominator is never zero, we conclude that
$\psi(\critprob)$ is continuous in~$\critprob$  and there is a value $\critprob\in(\cliquesize^{-3},\lambda/\cliquesize)$
such that $\psi(\critprob)=1/\gamma$. This gives~(\ref{eq:balance}). 
Note that the lower bound $\cliquesize^{-3}$ for~$\critprob$ is very
crude, but this is all that we will need.

First, to show that $\psi(\critprob)$ is monotonically decreasing as a function of $\critprob$, 
note from
Lemma~\ref{lem:coupleRC}
that $\Pr(\rv(\edgesubset)=\terminals)$ is monotonically decreasing in~$\critprob$ and
  $\Pr(\rv(\edgesubset)=1)$ is monotonically increasing in~$\critprob$.

Next, to show that $\psi(\cliquesize^{-3})> 1/\gamma$, we will assume $\critprob=\cliquesize^{-3}$ 
(and that $\cliquesize$ is sufficiently large)
and we will show
$\Pr(\rv(\edgesubset)=\terminals)> 1/(1+\gamma)$, which suffices.
By Lemma~\ref{lem:coupleone}, 
$\Pr(\rv(\edgesubset)=\terminals)\geq \Pr(\rv(\esu)=\terminals)$ where
$\esu$ is    drawn from $\ER(\gadget;\edgeprob)$.
Now the probability that 
$\esu \cap \clique^{(2)}=\emptyset$ is at least $1-\binom{\cliquesize}{2}
\cliquesize^{-3}$. In this case,
the probability that a particular
pair of terminals is connected in~$\esu \setminus \terminalset^{(2)}$ 
is at most $\cliquesize\times \cliquesize^{-3/4}\times \cliquesize^{-3/4}=\cliquesize^{-1/2}$, 
and the probability that there exists a connected pair  is at most
$\terminals^2\cliquesize^{-1/2}\leq \cliquesize^{-3/8}$.  
So $\Pr(\rv(\esu)=\terminals)  \geq1
- \binom{\cliquesize}{2}
\cliquesize^{-3}
-\cliquesize^{-3/8}>1/(1+\gamma)$
(for sufficiently large~$N$).

To finish the proof of~(\ref{eq:balance}),
we will show that $\psi(\lambda/\cliquesize)< 1/\gamma$. To do this, we will assume $\critprob=\lambda/\cliquesize$
(and that $\cliquesize$ is sufficiently large)
and we will show $\Pr(\rv(\edgesubset)=1)> \gamma/(1+\gamma)$, which suffices.
We will again use stochastic domination to compare the random cluster model to the 
Erd\H os-R\'enyi model, but this time we need some conditioning.

First,  we will  show that the probability that the graph
$(\clique,\edgesubset\cap \clique^{(2)})$ has a connected component of size at least $\theta \cliquesize$ 
tends to~$1$ as $\cliquesize \rightarrow\infty$.
To do this, let $\edgesubset^*$ be drawn from the distribution $\RC(\Gamma,q,\hat\edgeprob)$
where 
  $$\hat\edgeprob(e) = \left\{
\begin{array}{cc}
\critprob, &\mbox{if $e\in \clique^{(2)}$,}\\
0, &\mbox{otherwise.}
\end{array}\right.
$$ 
Since $\hat\edgeprob(e)\leq \edgeprob(e)$, Lemma~\ref{lem:coupleRC}
guarantees that the probability in question
is at least  the probability that the graph
$(\clique,\edgesubset^*)$ has a connected component of size at least $\theta \cliquesize$.
By Lemma~\ref{lem:clique}, this probability  tends to~$1$ as $\cliquesize\rightarrow \infty$.

Next, we will consider 
the generation of configuration~$\edgesubset$ from the distribution
$\RC(\gadget;q,\edgeprob)$ as follows. First, we will select a set $\edgesubset^+\subseteq \clique^{(2)}$ from
the appropriate induced distribution.
Letting  
  $\edgesubset^-=\clique^{(2)}\setminus \edgesubset^+$, we will select $\edgesubset$ from the distribution
$\RC(\gadget;q,\edgeprob;\edgesubset^+,\edgesubset^-)$.
We will finish by showing that, as long as $(\clique,\edgesubset^+)$ has a connected component of size 
at least $\theta \cliquesize$, 
$\Pr(\rv(\edgesubset)=1 \mid \edgesubset^+\subseteq \edgesubset, \edgesubset\cap \edgesubset^-=\emptyset)$
is greater than $\gamma/(1+\gamma)$.  
To do this, let $\edgeprob'(e) = \edgeprob(e)/q$ and let $\esl$ be a random variable drawn from
$\ER(\gadget;\edgeprob';\edgesubset^+,\edgesubset^-)$.
By Lemma~\ref{lem:coupletwo}, 
$\Pr(\rv(\edgesubset)=1 \mid \edgesubset^+\subseteq \edgesubset, \edgesubset\cap \edgesubset^-=\emptyset)\geq \Pr(\rv(\esl)=1)$.
Now in $\esl$, 
the probability that a particular terminal is not
connected to the large component is at most
$(1-q^{-1}\cliquesize^{-3/4})^{\theta \cliquesize}
\leq\exp(-\theta q^{-1}\cliquesize^{1/4})$,
and the probability that there exists a terminal that is not connected to this component is at most
$t\exp(-\theta q^{-1}\cliquesize^{1/4})$.   
So $\Pr(\rv(\esl)=1)>\gamma/(1+\gamma)$
(for sufficiently large~$N$). \end{proof}

\begin{lemma}\label{lem:excludedmiddletwo}
Fix $q>2$ and 
let $\lambda = \lambda_c + (q-\lambda_c)/2$.
Fix a weight $\gamma>0$ and let $\cliquesize_0$ be a sufficiently
large quantity depending on $q$ and $\gamma$.
Suppose a number of terminals $\terminals>1$ and a tolerance $0<\tol<1$
are given and fix 
$\cliquesize\geq\max\{\terminals^{16},\tol^{-1/8},\cliquesize_0\}$.
For every value of~$\critprob$
in the range $[\cliquesize^{-3},
\lambda/\cliquesize
]$,
 if $\edgesubset$ is drawn from $\RC(\gadget;q,p)$ then
\begin{equation}\label{eq:dichotomy}
\Pr(1<\rv(\edgesubset)<\terminals)<\tol.
\end{equation}
\end{lemma}

\begin{proof}
Let $C_1,C_2,\ldots$ be the connected 
components of $\edgesubset \cap \clique^{(2)}$, ordered in non-increasing size.
We are going to   rely  on the phase transition of the
random cluster model.
The main fact that we will use is that $|C_1|$ is likely to 
either be very small (around order $\log(\cliquesize)$)
or very large (a constant fraction of $\cliquesize$).
Thus, it is unlikely that $|C_1|$ is close to $\cliquesize^{1/8}$.
We will not need much detail about the phase transition.
We will show
\begin{equation}
 \Pr\big((|C_1|\leq \cliquesize^{1/8})\wedge(\rv(\edgesubset)<\terminals)\big)\leq \tol/2,
\label{eq:first}
\end{equation}
and
\begin{equation}
\label{eq:second}
\Pr\big((|C_1|> \cliquesize^{1/8})\wedge(\rv(\edgesubset)>1)\big) \leq \tol/2.
\end{equation}
(Actually, the event mentioned in (\ref{eq:second}) 
holds with all but exponentially small 
probability.)  The required inequality~(\ref{eq:dichotomy})
follows directly from (\ref{eq:first}) and~(\ref{eq:second}).

Inequality~(\ref{eq:first}) is easier.  To generate
a configuration~$\edgesubset$ from 
$\RC(\gadget;q,\edgeprob)$  we first select a set $\edgesubset^+\subseteq \clique^{(2)}$ from
the appropriate induced distribution. Then,
letting  
  $\edgesubset^-=\clique^{(2)}\setminus \edgesubset^+$, we   select $\edgesubset$ from the distribution
$\RC(\gadget;q,\edgeprob;\edgesubset^+,\edgesubset^-)$.
The size of $C_1$ is entirely determined by~$\edgesubset^+$.
If $\edgesubset^+$ has a component of size greater than $\cliquesize^{1/8}$
then 
$$\Pr((|C_1|\leq \cliquesize^{1/8})\wedge(\rv(\edgesubset)<\terminals) \mid 
\edgesubset^+\subseteq \edgesubset, \edgesubset\cap \edgesubset^-=\emptyset
 )
=0.$$
Otherwise, this probability is equal to
$ \Pr(\rv(\edgesubset) < \terminals \mid 
\edgesubset^+\subseteq \edgesubset, \edgesubset\cap \edgesubset^-=\emptyset
 )$.
By Lemma~\ref{lem:coupleone}, this is at most
$\Pr(\rv(\esu) < \terminals)$
where  $\esu$ is generated from
$\ER(\gadget;\edgeprob;\edgesubset^+,\edgesubset^-)$.

So now consider $\esu$. 
Fix attention on two terminals,
and colour the vertices in $\clique$ that are adjacent to the first terminal yellow,
and those adjacent to the second terminal blue.  We are in the situation
of Lemma~\ref{lem:yellowblue} 
with $\nu=\cliquesize$, $\nu_{\max}\leq\cliquesize^{1/8}$
and $\sprob=\cliquesize^{-3/4}$.  
Thus the probability that there exists a bicoloured block 
(i.e., that the terminals are connected via some connected component~$C_j$)
is at most 
$\cliquesize[1-(1-\cliquesize^{-3/4})^{\cliquesize^{1/8}}]^2
\leq \cliquesize[1-1+\cliquesize^{-5/8}]^2=\cliquesize^{-1/4}$.
Thus the probability that there exists a pair of connected terminals is 
at most $\binom{\terminals}{2}\cliquesize^{-1/4}\leq\cliquesize^{-1/8}/2<\tol/2$.
Note that the event ``there exists a pair of connected terminals'' is
the same as the event $\rv(\esu)<\terminals$.

At a high level, the path to establishing inequality~(\ref{eq:second})
is as follows.  We define events $\largec$ (``large component'')
and $\weighty$ (``weighty components'') and establish
\begin{align}
\Pr\big((|C_1|>\cliquesize^{1/8})\wedge \neg\largec\big)&<\tol/6,\label{eq:bd1}\\
\Pr(\largec\wedge\neg\weighty)&<\tol/6,\quad\text{and}\label{eq:bd2}\\
\Pr\big(\weighty\wedge(\rv(\edgesubset)>1)\big)&<\tol/6.\label{eq:bd3}
\end{align}
Then the required inequality~(\ref{eq:second})
follows by elementary algebra of sets (events).

We will say that an event holds ``with high probability'' 
(abbreviated whp)
if the probability that it holds is 
$o(1/n^k)$ for any fixed constant~$k$. 

First, we prove inequality~(\ref{eq:bd1}).
Let $S_1,S_2,\ldots$ be the connected components of $\edgesubset$
and let $\Chat_1,\Chat_2,\ldots$ be the sets
$S_1 \cap \clique^{(2)},S_2 \cap \clique^{(2)},\ldots$, ordered in non-increasing size.
(Thus, $\{\Chat_j\}$
is a coarsening of $\{C_j\}$.)  Event $\largec$ is 
$|\Chat_1|>\cliquesize^{15/16}$.
Let $\mathcal{E}_1$ be the event 
$(|C_1|>\cliquesize^{1/8})\wedge \neg\largec$.
We are interested showing that the 
event $\mathcal{E}_1$   is unlikely.

Construct the red subgraph and green subgraph 
of $(\gadgetvertices,\edgesubset)$
as in Lemma~\ref{lem:FL}
with $r=1/q$.
Let $\mathcal{E}_2$ be the event
$$(|R\cap \clique|\leq \cliquesize/q+ \cliquesize^{15/16} )\wedge (C_1\subseteq R) \wedge (|C_1|>\cliquesize^{1/8})
 .$$
We will show
$$\Pr(\mathcal{E}_2)\geq\Pr(\mathcal{E}_1)/q^2$$
and
$$\Pr(\mathcal{E}_2)  \leq \eta/(6 q^2),$$
which proves~(\ref{eq:bd1}).
The first of these follows from $\Pr(\mathcal{E}_2 \mid \mathcal{E}_1)\geq q^{-2}$,
which follows since
$\Pr(C_1\subseteq R \mid \mathcal{E}_1) = q^{-1}$ (by the definition of the red subgraph)
and  
$\Pr\bigl(|R \cap \clique|\leq \cliquesize/q+ \cliquesize^{15/16}\bigm|
\mathcal{E}_1 \wedge (C_1\subseteq R)
 \bigr)\geq q^{-1}$
(also by the definition of the red subgraph --- with probability $q^{-1}$, red is the rarest colour among
 vertices in $\clique\setminus \Chat_i$ where $\Chat_i$ is the element of $\{\Chat_j\}$ containing $C_1$).

Now the probability of event $\mathcal{E}_2$
is at most 
$$ 
\Pr\big( (C_1\subseteq R) \wedge (|C_1|>\cliquesize^{1/8}) \bigm|  
|R\cap \clique|\leq \cliquesize/q+ \cliquesize^{15/16}\big)$$
which is at most
the conditional probability that 
 $\edgesubset[R\cap\clique]$ has a component of size at least
 $\cliquesize^{1/8}$,
conditioned on   $|R\cap \clique|\leq \cliquesize/q+ \cliquesize^{15/16}$.
But, by Lemma~\ref{lem:FL},
$\edgesubset[R\cap\clique]$ is 
an Erd\H os-R\'enyi
random graph with   
$|R\cap\clique|\leq(q^{-1}+\cliquesize^{-1/16})\cliquesize$ vertices and
subcritical edge probability $\critprob \leq \lambda/\cliquesize<|R\cap\clique|^{-1}$.
So the probability that it has a  
component of 
size at least $\cliquesize^{1/8}$ is at most
$\epsilon(\cliquesize)$, where $\epsilon(\cdot)$ is smaller than 
any inverse polynomial \cite[Proof of Theorem~5.4]{JLR}.
Thus, we have shown $\Pr(\mathcal{E}_2)  \leq \eta/(6 q^2)$
(provided $\cliquesize$ is sufficiently large), so we have proved~(\ref{eq:bd1}).

Next we prove inequality~(\ref{eq:bd2}).  
The event $\weighty$ is that $|C_1|+\cdots+|C_\fewcomponents|\geq 
\cliquesize^{15/16}$, for some $\fewcomponents\leq2\cliquesize^{5/16}$.
The component $\Chat_1$
is in general composed of a number of components from~$\{C_j\}$.
However, the number of constituent components cannot be larger than
the number $|\es\cap(\clique\times\terminalset)|$ of edges joining 
$\clique$ to $\terminalset$.  This number can be upper-bounded using a
bounding configuration~$\esu$ drawn from $\ER(\gadget;\edgeprob)$.
We see from a Chernoff bound that whp the number 
$|\esu\cap(\clique\times\terminalset)|$ of edges joining
$\clique$ to $\terminalset$ is 
less than $2\terminals\cliquesize^{1/4}\leq2\cliquesize^{5/16}$,
i.e., twice the expected number,   
so whp,
at most $2\cliquesize^{5/16}$
 components from~$\{C_j\}$  contribute to $\Chat_1$.
Thus $\Pr(\weighty\mid\largec)=1-\epsilon(\cliquesize)$
and (\ref{eq:bd2}) follows.

Finally we prove inequality~(\ref{eq:bd3}).
As before, generate~$\edgesubset$ by selecting a
set $\edgesubset^+\subseteq \clique^{(2)}$ from
the appropriate induced distribution,
letting  
  $\edgesubset^-=\clique^{(2)}\setminus \edgesubset^+$,  and selecting $\edgesubset$ from the distribution
$\RC(\gadget;q,\edgeprob;\edgesubset^+,\edgesubset^-)$.
We will condition on the event $\weighty$ (which is entirely determined by $\edgesubset^+$).
As above, 
$$\Pr\big(\rv(\edgesubset)=1\bigm| \edgesubset^+\subseteq \edgesubset, \edgesubset\cap \edgesubset^-=\emptyset\big)
\geq \Pr(\rv(\esl)=1),$$
where
$\edgeprob'(e) = \edgeprob(e)/q$ and
$\esl$ is a random variable drawn from
$\ER(\gadget;\edgeprob';\edgesubset^+,\edgesubset^-)$.
Consider two terminals $i,j$ in $\terminalset$.
We can use Lemma~\ref{lem:yellowblue} to  find an upper bound for the probability that 
$i$ and $j$ are {\it not\/} connected 
in $\esl$
via one of the components $C_1,\ldots,
C_\fewcomponents$.  In $\esl$, edges from $\terminalset$ to $\clique$
are selected independently with probability $q^{-1}\cliquesize^{-3/4}$. 
If we colour vertices in $\clique$ adjacent to $i$ (respectively, $j$)
yellow (respectively,  blue) then we are in the situation 
of Lemma~\ref{lem:yellowblue}, with $\sprob=q^{-1}\cliquesize^{-3/4}$.
From the remarks about monotonicity in $s$ and~$\nu$, we may assume
for an upper bound that $\nu=\cliquesize^{15/16}$ 
and $\fewcomponents=2\cliquesize^{5/16}$.  
The probability we want to bound is that 
of not having a bicoloured component.   
Observe 
$$(1-\sprob)^{\nu/\fewcomponents}=(1-q^{-1}\cliquesize^{-3/4})^{\frac12\cliquesize^{5/8}}
\leq 1-
\tfrac14
q^{-1}\cliquesize^{-1/8},$$
for $\cliquesize$ sufficiently large,
and since the function $u\mapsto u(2-u)$ is increasing in the range $(0,1)$,
$$(1-\sprob)^{\nu/\fewcomponents}(2-(1-\sprob)^{\nu/\fewcomponents})
\leq (1-
\tfrac14
q^{-1}\cliquesize^{-1/8})(1+
\tfrac14
q^{-1}\cliquesize^{-1/8}).$$
Applying Lemma~~\ref{lem:yellowblue},
$$\Pr(\text{$i\not\sim j$ in $\esl$} )
\leq
(1-
\tfrac1{16}
q^{-2}\cliquesize^{-1/4})^{2\cliquesize^{5/16}}
<\exp(-
\tfrac18
q^{-2}\cliquesize^{1/16}).$$
So whp $i\sim j$.  It follows that whp all vertices in
$\terminalset$ are connected to each other, and 
so $\Pr(\rv(\esl)=1 )\geq1-\epsilon(\cliquesize)$.
This deals with~(\ref{eq:bd3}) and completes the proof.
\end{proof}

\subsection{The random cluster model on a clique connected to an independent set} 
\label{sec:variant}

Construct~$\gadget$ as in Section~\ref{sec:gadget}.
Let $\vargadget=(\gadgetvertices,\gadgetedges \setminus\terminalset^{(2)})$ be
the graph derived from~$\gadget$ by deleting edges within~$\terminalset$.
The vertices in~$\terminalset$ are called the ``terminals'' of $\vargadget$.
Let  $\boldgamma=\{\gamma_\graphedge\}$ be the set of edge weights
defined by $\gamma_\graphedge = \edgeprob(\graphedge)/(1-\edgeprob(\graphedge))$.

For an edge subset $\edgesubset'\subseteq\gadgetedges\setminus\terminalset^{(2)}$,
let $\kappa'(\gadgetvertices,\edgesubset')$ denote the number of connected components
that do not contain terminals 
in the graph $(\gadgetvertices,\edgesubset')$.
Let $\calA^k$ denote the set of edge subsets $\edgesubset'\subseteq\gadgetedges
\setminus\terminalset^{(2)}$
for which the terminals of $(\gadgetvertices,\edgesubset')$ 
are contained in exactly $k$~connected components.
Let $\calA = \bigcup_{k\in[t]} \calA^k$ (this is the set of all edge subsets of $\vargadget$).
Let $Z^k$ be $q^{-k}$ times the contribution to $\ZTutte(\vargadget;q,\boldgamma)$ from 
edge subsets $A'\in \calA^k$.
Formally, $Z^k = \sum_{A' \in \calA^k} q^{\kappa'(V,A')}\gamma({A'})$.
Let $Z = \sum_{k=1}^t Z^k$.

We will use the following lemma to apply
Lemmas~\ref{lem:excludedmiddleone}
and~\ref{lem:excludedmiddletwo}
from Section~\ref{sec:gadget}
in our reductions.

\begin{lemma}
\label{lem:wiring}
$Z^k/Z=\Pr(\rv(\edgesubset)=k)$, where
$\edgesubset$ is drawn from $\RC(\gadget;q,p)$.
\end{lemma}
\begin{proof}

From the definitions of the random cluster model,
$$\Pr(\rv(\edgesubset)=k) = 
\frac{\sum_{A'\in \calA^k} \tildeP(\gadget;A'\cup T^{(2)},q,\edgeprob)}
{
\sum_{A'\in \calA} \tildeP(\gadget;A'\cup T^{(2)},q,\edgeprob)}.$$

Plugging in the definition of $\tildeP(\gadget;A'\cup T^{(2)},q,\edgeprob) $,
this is
$$
\frac{\sum_{A'\in \calA^k} \gamma({A'})q^{\kappa'(\gadgetvertices,\edgesubset')+1} \prod_{e\in \gadgetedges \setminus T^{(2)}}(1-\edgeprob(e))}
{\sum_{A'\in \calA} \gamma({A'})q^{\kappa'(\gadgetvertices,\edgesubset')+1}
\prod_{e\in \gadgetedges \setminus T^{(2)}}(1-\edgeprob(e))},$$
which is what we require, once we cancel a factor of 
$$q \prod_{e\in \gadgetedges \setminus T^{(2)}}(1-\edgeprob(e))$$ from the numerator and denominator.
\end{proof}
 
\section{Computational problems, fully polynomial randomised approximation schemes and efficiently approximable real numbers}
\label{sec:FPRAS}
 
Fix real numbers $q>2$ and $\gamma>0$ and consider the
following computational problem, which is parameterised by~$q$ and~$\gamma$.

\begin{description}
\item[Problem] $\Tutte(q,\gamma)$.
\item[Instance] Graph $\graph=(\graphvertices,\graphedges)$.
\item[Output]  $\ZTutte(\graph;q,\boldgamma)$,
where
$\boldgamma$ is the constant function  with $\boldgamma_\graphedge = \gamma$ for every $\graphedge\in\graphedges$.
\end{description}

We will have much more to say about computational approximations of real numbers below. For the moment, it may help 
the reader to think of~$q$ and~$\gamma$ as being rational.
Jaeger, Vertigan and Welsh~\citeyear{JVW90}
have shown that  $\Tutte(q,\gamma)$ is \#P-hard
for every fixed $q>2$ and $\gamma>0$.
Thus, it is unlikely that there is a polynomial-time algorithm for exactly
solving this problem. (If there were such an algorithm, this would entail 
$\Ptime=\numP$, and of course $\Ptime=\mathrm{NP}$.)

We are interested in the complexity of \emph{approximately} solving $\Tutte(q,\gamma)$.
We start by defining the relevant concepts.
A \emph{randomised approximation scheme\/} is an algorithm for
approximately computing the value of a function~$f:\alphabet^*\rightarrow
\mathbb{R}$.
(Here, $\alphabet$ is a finite alphabet, and inputs to~$f$ are represented as
strings over this alphabet.)
The
approximation scheme has a parameter~$\varepsilon>0$ which specifies
the error tolerance.
A \emph{randomised approximation scheme\/} for~$f$ is a
randomised algorithm that takes as input an instance $ x\in
\alphabet^{\ast }$ (e.g., for the problem $\Tutte(q,\gamma)$, the
input would be an encoding of a graph~$\graph$) and a rational error
tolerance $\varepsilon >0$, and outputs a rational number $z$
(a random variable of the ``coin tosses'' made by the algorithm)
such that, for every instance~$x$,
\begin{equation}
\label{eq:3:FPRASerrorprob}
\Pr \big[e^{-\epsilon} f(x)\leq z \leq e^\epsilon f(x)\big]\geq \frac{3}{4}\, .
\end{equation}
The randomised approximation scheme is said to be a
\emph{fully polynomial randomised approximation scheme},
or \emph{FPRAS},
if it runs in time bounded by a polynomial
in $ |x| $ and $ \epsilon^{-1} $.

Note that the quantity $\frac34$ in
Equation~(\ref{eq:3:FPRASerrorprob})
could be changed to any value in the open
interval $(\frac12,1)$ without changing the set of problems
that have randomised approximation schemes \cite[Lemma~6.1]{jvv}.
There is another sense in which FPRAS is a robust notion of approximability,
namely that the existence of a polynomial-time algorithm that achieves constant factor approximations
often implies the existence of an FPRAS\null. We are not aware of a reference to this observation in the
published literature, but it is easily explained in the context of the specific problem $\Tutte(q,\gamma)$.
For any graph $\graph$, denote by $k\cdot \graph$ the graph composed of $k$ disjoint copies of~$\graph$.
Then $\ZTutte(k\cdot\graph;q,\boldgamma)=\ZTutte(\graph;q,\boldgamma)^k$.  So, setting $k=O(\epsilon^{-1})$,
a constant factor approximation to $\ZTutte(k\cdot\graph;q,\boldgamma)$ will yield (by taking the $k$th root)
an FPRAS for $\ZTutte(\graph;q,\boldgamma)$.  Clearly, an approximation within a polynomial factor would 
also suffice.

We say that a real number~$z$ is \emph{efficiently approximable} if there is an FPRAS
for the problem 
of computing~$z$. (Technically, we can view the
problem of computing~$z$ as a degenerate computational problem in which 
every input gets mapped to the output~$z$.)
Approximations to real numbers are useful.
For example, if $\hat q$ and $\hat \gamma$ are
approximations to~$q$ and~$\gamma$ satisfying
$$e^{-\tfrac{\varepsilon}{n+m}} q \leq \hat q \leq 
e^{\tfrac{\varepsilon}{n+m}} q$$
and
$$e^{-\tfrac{\varepsilon}{n+m}} \gamma \leq \hat \gamma \leq 
e^{\tfrac{\varepsilon}{n+m}}
\gamma$$
and $\hat\boldgamma_\graphedge = \hat\gamma$ for every $\graphedge\in\graphedges$
then
\begin{equation}
\label{approxparams}
e^{-\varepsilon }
\ZTutte(\graph;q,\boldgamma) \leq
\ZTutte(\graph;\hat q,\hat \boldgamma) \leq 
e^{\varepsilon }
\ZTutte(\graph; q, \boldgamma).
\end{equation} 
Thus, to approximate $\ZTutte(\graph;q,\boldgamma)$,
it suffices to first compute rational approximations $\hat{q}$ and $\hat \gamma$, and then approximate 
$\ZTutte(\graph;\hat q,\hat \boldgamma)$.
The reason that our definition of an efficiently
approximable real number~$z$ allows $z$ to be approximated by an
FPRAS rather than, for example, by a deterministic approximation
scheme, is that the more general definition allows us to obtain more general
results, which will be easier to combine with other results in the area.

When the parameters are efficiently approximable
reals, it is possible
to approximate quantities associated with the gadgets~$\gadget$
and~$\vargadget$ 
that we studied in Section~\ref{sec:gadget} and~\ref{sec:variant}.

\begin{lemma}
\label{lem:computeZ}
Suppose $q>2$ is an efficiently  approximable real.
Consider the gadget $\vargadget$ from Section~\ref{sec:variant} with
parameters~$\terminals$, $\cliquesize$ and $\critprob$  
where 
$\critprob\in[0,1]$ is a rational number and $\cliquesize^{1/4}$ is an integer.
There is an FPRAS for computing $Z^1$ and $Z^\terminals$, given
inputs~$\terminals$, $\cliquesize$ and~$\critprob$.
\end{lemma}

\begin{proof}
In this proof only, the notation $\vargadget_{\cliquesize,\terminals}$ is used to make explicit the 
size of the gadget~$\vargadget$. 

Let $E(\vargadget_{\cliquesize,\terminals})$ be the set of edges of $\vargadget_{\cliquesize,\terminals}$.
Let $\calA^{k,\ell}$ denote the set of edge subsets  $A\subseteq E(\vargadget_{\cliquesize,\terminals})$
 with $k$~connected components containing vertices in $\terminalset$, and $\ell$
other connected components.
Let $w(\terminals,\cliquesize,k,\ell) = \sum_{A\in \calA^{k,\ell}} \gamma(A)$.
Thus
$Z^k=\sum_{\ell=0}^\cliquesize w(\terminals,\cliquesize,k,\ell)q^\ell$.  We exhibit recurrence
relations for $w(\terminals,\cliquesize,j,k)$, from which it follows that 
these
can be computed in
polynomial time by dynamic programming.  
The rest of this proof is straightforward, but provides the details.
To reduce the number of boundary cases
we need to consider, it is convenient to allow one or other of
$k$ and $\ell$ to take on the value~$-1$.
Of course, we stipulate 
$$
w(\terminals,\cliquesize,-1,\ell)=w(\terminals,\cliquesize,k,-1)=0, \quad
   \text{for all $\terminals, \cliquesize,k,\ell\geq0$}.
$$
Another easy-to-verify boundary case is
$$
w(\terminals,0,k,\ell)=\begin{cases}1,&\text{when $k=\terminals$ and $\ell=0$};\\
   0,&\text{otherwise},
\end{cases}
$$
which is valid for all $\terminals,k,\ell\geq0$. 
Also  $w(\terminals,\cliquesize,k,\ell)=0$ if exactly one of $\terminals$ and $k$ is $0$
and  $w(0,\cliquesize,0,0)$ is~$1$ is $\cliquesize=0$ and~$0$ otherwise.
Finally, $w(0,0,k,\ell) = 0$ if $k+\ell>0$.

The general recurrence, covering all situations other than the boundary cases
already mentioned, can now be given.
\begin{align}
w(\terminals,\cliquesize,k,\ell)
   &=\sum_{\textstyle{{1\leq i\leq\terminals\atop1\leq j\leq\cliquesize}}}
      \binom{\terminals}{i}\binom{\cliquesize-1}{j-1}w(i,j,1,0)w(\terminals-i,\cliquesize-j,k-1,\ell)
      \label{eq:rec1}\\
   &\qquad\null+\sum_{1\leq j\leq\cliquesize}\binom{\cliquesize-1}{j-1}
      w(0,j,0,1)w(\terminals,\cliquesize-j,k,\ell-1).
    \label{eq:rec2}
\end{align}
(Although valid for all $k,\ell\geq0$, the recurrence becomes trivial 
when $k+\ell=1$, a point we must return to at the end.)
The explanation is as follows.  We partition the sum defining $w(\terminals,\cliquesize,k,\ell)$
according to the connected component~$C$ in $A$ containing some distinguished 
vertex $v$ in~$\clique$.  
(Note that $\cliquesize>0$, since one of the boundary cases covers $\cliquesize=0$.)
Let $i=|C\cap T|$ be the number of vertices of~$C$ that lie in~$T$, and $j=|C\cap K|$ 
the number that lie in~$K$.  Summation (\ref{eq:rec1}) deals with the situation $i>0$ and 
(\ref{eq:rec2}) with the situation $i=0$.  The binomial coefficients in~(\ref{eq:rec1})
count the number of
connected components $C\ni v$ that contain the distinguished vertex, and 
have the correct intersections $|C\cap\terminalset|=i$ and $|C\cap\clique|=j$
with $\terminalset$ and~$\clique$;  the factor $w(i,j,1,0)$ counts the weight of {\it connected\/} 
subgraphs of $\vargadget[C]$;  and $w(\terminals-i,\cliquesize-j,k-1,\ell)$
the weight of subgraphs of $\vargadget[\clique\cup\terminalset\setminus C]$ that 
have the correct number of connected components (i.e., $k-1$ with vertices in~$T$,
and $\ell$ without).  The analysis of summation~(\ref{eq:rec2}) is entirely similar.

Assuming $k+\ell>1$, the recurrence is well founded, in the sense that all occurrences
of $w(\terminals',\cliquesize',\cdot,\cdot)$ on the right hand side (apart from those that
get multiplied by zero) have 
$\terminals'+\cliquesize'<\terminals+\cliquesize$. 
When $k+\ell=1$, i.e., $(k,\ell)\in\{(0,1),(1,0)\}$ the recurrence (\ref{eq:rec1},\ref{eq:rec2})
becomes the trivial $w(\terminals,\cliquesize,k,\ell)=w(\terminals,\cliquesize,k,\ell)$.
Note that $(k,\ell)=(0,1)$  entails $\terminals=0$
(otherwise $w(\terminals, \cliquesize,k,\ell)=0$)
and $(k,\ell)=(1,0)$ entails $\terminals>0$,
so only one of  the possibilities $(k,\ell)\in\{(0,1),(1,0)\}$ occurs for a given pair
$(\terminals, \cliquesize)$.
In order to make progress 
in this situation, we employ a preprocessing step,
applying first the identity
$$
w(\terminals, \cliquesize,1,0)=\prod_{e\in E(\vargadget_{\cliquesize,\terminals})}(1+\gamma_e)
   -\sum_{\textstyle{0\leq k\leq \terminals,0\leq\ell\leq\cliquesize\atop k+\ell>1}}
   w(\terminals,\cliquesize,k,\ell)
$$
and then expanding the terms on the right hand side according to the usual recurrence.
Exactly the same formal expression applies to $w(\terminals, \cliquesize,0,1)$.  The identity
merely expresses complementation:  the weight of subgraphs with $k+\ell=1$ equals
the total weight of subgraphs, less the weight of subgraphs with $k+\ell>1$.
With this modification, the recurrence becomes well founded.

Since only a polynomial number of distinct tuples $(\terminals,\cliquesize,k,\ell)$ arise,
the recurrence can be solved by dynamic programming in polynomial time.
Note that the computation of $\prod_{e\in E(\vargadget_{\cliquesize,\terminals})}(1+\gamma_e)$
in the preprocessing step can be done exactly since $\critprob$ is a rational 
and $\cliquesize^{1/4}$ is an integer.

Once the $w(\terminals,\cliquesize,k,\ell)$ values are computed, we wish to estimate
$$Z^k=\sum_{\ell=0}^\cliquesize w(\terminals,\cliquesize,k,\ell)q^\ell.$$
Let $\varepsilon$ be the desired accuracy in the 
approximation of~$Z^k$.
Compute a rational $\hat q$ 
in the range
$e^{-\varepsilon/N} q \leq \hat q \leq e^{\varepsilon/N} q$ and return
$Z^k=\sum_{\ell=0}^\cliquesize w(\terminals,\cliquesize,k,\ell){\hat q}^\ell$.
\end{proof}

It will also be necessary for us to approximate the critical edge
probability~$\critprob$, so that, with this approximation,  the graph~$\gadget$  
approximately satisfies Equation~(\ref{eq:balance}) in Lemma~\ref{lem:excludedmiddleone}.
The following lemma shows that this is 
possible.
 
\begin{lemma}
\label{lem:computerho}
Suppose $q>2$ is an efficiently  approximable real. Fix $\gamma>0$ and 
 let $\lambda = \lambda_c + (q-\lambda_c)/2$.
Suppose that $\chi>0$ is rational. Consider the gadget $\gadget$ from Section~\ref{sec:gadget} with parameters~$\terminals$, $\cliquesize$, and $\critprob$. There is 
a randomised algorithm
whose running time is at most a polynomial in $\chi^{-1}$, $\cliquesize$ and $\terminals$ which takes input $\cliquesize$ and $\terminals$ (where it is assumed that $\cliquesize^{1/4}$ is an integer and that $\cliquesize\geq\max\{\terminals^{16}, \cliquesize_0\}$ for the constant $\cliquesize_0$ 
from Lemma~\ref{lem:excludedmiddleone}) 
and, with probability at least~$3/4$, computes
a rational ${\critprob}$ in the range $[\cliquesize^{-3},\lambda/\cliquesize]$ such that, if $\edgesubset$ is drawn from $\RC(\gadget;q,p)$, then \begin{equation} \label{eq:compbalance} e^{-\chi} \gamma \leq \frac{\Pr(\rv(\edgesubset)=1)} {\Pr(\rv(\edgesubset)=\terminals)} \leq e^{\chi} \gamma.
\end{equation}
\end{lemma}

\begin{proof}  
First, we establish a useful preliminary fact.
If $0<\critprob\leq 1/4$, $0<\delta\leq 1$
and 
$\frac{1}{ 1+\delta} \critprob \leq \hat \critprob \leq (1+\delta) \critprob$
then 
\begin{equation}
\label{eq:trivia}
e^{-2\delta} \frac{\critprob}{1-\critprob}
\leq \frac{\hat\critprob}{1-\hat\critprob}
\leq e^{2\delta} \frac{\critprob}{1-\critprob}.
\end{equation}
To see this, note that $e^{\delta} \critprob \leq 3/4$
so $1-e^{\delta} \critprob\geq  \critprob$. Then
letting $x = e^{\delta}-1$,
$$1-\critprob 
\leq  
1 -\critprob + x(1-e^{\delta}\critprob)   - x  \critprob
= (1+x)(1-(1+x)\critprob)
= e^{\delta}(1-e^{\delta} \critprob),$$
which, together with $1+\delta \leq e^{\delta}$, gives the right-most inequality in~(\ref{eq:trivia}).
Similarly, 
$\critprob e^{-\delta} \leq 1/4$ so $1-e^{-\delta} \critprob \geq
\critprob  $ so
letting $x=1-e^{-\delta}$,
$$1-\critprob 
\geq  
1 -\critprob + x \critprob   - x(1-e^{-\delta}\critprob)    
= (1-x)(1-(1-x)\critprob)
= e^{-\delta}(1-e^{-\delta} \critprob),$$
which, together with $e^{-\delta} \leq \frac{1}{1+\delta}$, gives the left-most inequality in~(\ref{eq:trivia}).
 
Now consider the gadget $\gadget$  with parameters~$\terminals$,
$\cliquesize$, and $\critprob$.
Let $n$ and $m$ be the number of vertices, and edges, respectively, of
the variant $\vargadget$.
We will be interested in three quantities which depend
upon~$\critprob$: $Z^1$, $Z^\terminals$ and 
$\zeta \doteq \frac{\Pr(\rv(\edgesubset)=1)}
{\Pr(\rv(\edgesubset)=\terminals)}$.
For this proof only, we will denote these as $Z^1_\critprob$,
$Z^\terminals_\critprob$ and $\zeta_\critprob$ to make the dependence on
$\critprob$ clear.
By Lemma~\ref{lem:wiring}, $\zeta_\critprob = Z^1_\critprob/Z^\terminals_\critprob$.

Lemma~\ref{lem:computeZ} provides an FPRAS for computing $\zeta_\critprob$ for
rational parameters~$\critprob$. 
We will 
set $\delta = \chi/(16(n+m))$ and we will
use the FPRAS with accuracy parameter~$\chi/4$ to estimate
$\zeta_\critprob$ for every rational
$$\critprob = \frac{1}{\cliquesize^3} {(1+\delta)}^\mu$$
where $\mu$ is a non-negative integer
with $  \frac{1}{\cliquesize^3} {(1+\delta)}^\mu \leq \frac{\lambda}{\cliquesize}$.
Note that the number of integers~$\mu$ in this range is $O(\delta^{-1}
\log \cliquesize)$.
We will first ``power up'' the success probability of the FPRAS using standard techniques \cite{jvv} so
that the probability that any of these $O(\delta^{-1}
\log \cliquesize)$ calls to the FPRAS fails is at most~$3/4$.
If we find a rational $\hat \critprob$ for which our
estimate $\hat \zeta_{\hat \critprob}$ satisfies
$e^{-\chi/2} \gamma \leq  \hat \zeta_{\hat \critprob} \leq e^{\chi/2}
\gamma$, we will return this value~$\hat\critprob$.
Clearly, this will be an acceptable answer, since
$e^{-\chi/4} \zeta_{\hat\critprob} \leq  
\hat \zeta_{\hat \critprob} \leq
e^{\chi/4} \zeta_{\hat\critprob}$.

To finish the proof, we just need to argue that we will try a
$\hat\critprob$ which we accept.
Lemma~\ref{lem:excludedmiddleone}
guarantees that there is a value $\critprob^*\in
[\cliquesize^{-3},\lambda \cliquesize^{-1}]$ such that
$\zeta_{\critprob^*} = \gamma$.
In our computations, we will compute $\zeta_{\hat \critprob}$ for some
$\hat \critprob$ in the range
$ \frac{1}{1+\delta} \critprob^* \leq \hat \critprob \leq (1+\delta) \critprob^*$.
Recall that every edge~$e$ of
the gadget has edge-weight $\gamma_e=p(e)/(1-p(e))$,
and that $p(e)=\critprob$ for edges within the part of the gadget
corresponding to the clique~$K$. The weights of the other edges do not
depend on $\critprob$.
Thus, by~(\ref{eq:trivia}) 
and by analogy to~(\ref{approxparams}), 
$e^{-\chi/4} \zeta_{\critprob^*} \leq \zeta_{\hat\critprob} \leq
e^{\chi/4} \zeta_{\critprob^*}$. Thus, $\hat \critprob$ will be accepted.
\end{proof}

\section{Approximation-preserving reductions and   \BIS}  
\label{sec:APred}

Our main tool for understanding the relative difficulty of
approximation counting problems is \emph{approximation-preserving reductions}.
We use
Dyer, Goldberg, Greenhill and Jerrum's notion of
approximation-preserving reduction \cite{APred}.
Suppose that $f$ and $g$ are functions from
$\alphabet^{\ast }$ to~$\mathbb{R}$. An ``approximation-preserving
reduction'' from~$f$ to~$g$ gives a way to turn an FPRAS for~$g$
into an FPRAS for~$f$. Here is the definition. An {\it approximation-preserving reduction\/}
from $f$ to~$g$ is a randomised algorithm~$\mathcal{A}$ for
computing~$f$ using an oracle for~$g$. The algorithm~$\mathcal{A}$ takes
as input a pair $(x,\varepsilon)\in\alphabet^*\times(0,1)$, and
satisfies the following three conditions: (i)~every oracle call made
by~$\mathcal{A}$ is of the form $(w,\delta)$, where
$w\in\alphabet^*$ is an instance of~$g$, and $0<\delta<1$ is an
error bound satisfying $\delta^{-1}\leq\poly(|x|,
\varepsilon^{-1})$; (ii) the algorithm~$\mathcal{A}$ meets the
specification for being a randomised approximation scheme for~$f$
(as described above) whenever the oracle meets the specification for
being a randomised approximation scheme for~$g$; and (iii)~the
run-time of~$\mathcal{A}$ is polynomial in $|x|$ and
$\varepsilon^{-1}$.

If an approximation-preserving reduction from $f$ to~$g$
exists we write $f\APred g$, and say that {\it $f$ is AP-reducible
  to~$g$}.
Note that if $f\APred g$ and $g$ has an FPRAS then $f$ has an FPRAS\null.
(The definition of AP-reduction was chosen to make this true).
If $f\APred g$ and $g\APred f$ then we say that
{\it $f$ and $g$ are AP-interreducible}, and write $f\APeq g$.

The definitions allow us to construct approximation-preserving
reductions between problems~$f$ and~$g$ with real parameters without
insisting that the parameters themselves be efficiently approximable. Nevertheless,  some
of our results restrict attention to efficiently approximable parameters.
According to the definition, approximation-preserving reductions may use randomisation. Nevertheless,
the reductions that we present in this paper are deterministic,
except where they make use of an FPRAS to approximate a real parameter.
A word of warning about terminology:
Subsequent to \cite{APred}, the notation $\APred$ has been
used
to denote a different type of approximation-preserving reduction which applies to
optimisation problems.
We will not study optimisation problems in this paper, so hopefully this will
not cause confusion.

Dyer et al.~\citeyear{APred} studied counting problems in \#P and
identified three classes of counting problems that are interreducible
under approximation-preserving reductions. The first class, containing the
problems that admit an FPRAS, are trivially AP-interreducible since
all the work can be embedded into the reduction (which declines to
use the oracle). The second class  is the set of problems that are
AP-interreducible with \SAT, the problem of counting
satisfying assignments to a Boolean formula in CNF\null.
Zuckerman~\citeyear{zuckerman}
has shown that \SAT{} cannot have an FPRAS unless
$\mathrm{RP}=\mathrm{NP}$. The same is obviously true of any problem
 to which \SAT{} is AP-reducible.  
 
The third class appears to be of intermediate complexity.
It contains   all of the counting problems
expressible in a certain logically-defined complexity class. Typical
complete problems include counting the downsets in a partially ordered
set \cite{APred},
computing the partition function of the ferromagnetic Ising model with
varying interaction energies and local external magnetic
fields \cite{ising}
and counting the independent sets in a bipartite graph,
which is defined as follows.

\begin{description}
\item[Problem] $\BIS$.
\item[Instance] A bipartite graph $B$.
\item[Output]  The number of independent sets in $B$.
\end{description}

We showed in \cite{APred} that \BIS\ is complete for the
logically-defined
complexity class $\mathrm{\#RH}\Pi_1$  with respect to approximation-preserving
reductions.
We presume that there is no FPRAS for \BIS, but this is not
known.

\section{Our results}
\label{sec:results}
The next three sections of this paper give an approximation-preserving
reduction from \BIS\ to the problem $\Tutte(q,\gamma)$.
 
We start by defining the problem of computing the Tutte polynomial of 
a uniform hypergraph~$H$
with fixed positive edge weights.
For fixed positive real numbers~$q$ and $\gamma$ 
the problem $\uhTutte(q,\gamma)$ is defined as follows.
 
\begin{description}
\item[Problem] $\uhTutte(q,\gamma)$.
\item[Instance]  A uniform hypergraph $\hypergraph=(\hypervertices,\hyperedges)$.
\item[Output]  $\ZTutte(\hypergraph;q,\boldgamma)$,
where $\boldgamma$ is the constant function  with $\boldgamma_\hyperedge = \gamma$ for every $\hyperedge\in\hyperedges$.
\end{description}

In Section~\ref{bistohypertutte} 
we show that, for every $q>1$, there is an approximation-preserving
reduction from \BIS\ to $\uhTutte(q,q-1)$.

$\uhTutte(q,\gamma)$ is the problem of computing the Tutte polynomial of
a uniform hypergraph.
The most difficult part of the paper is reducing this approximation problem
to the problem of approximately computing the (multivariate) Tutte polynomial of an undirected
graph (a $2$-uniform hypergraph).
Section~\ref{sec:tographs}
shows that for any positive real numbers $q>2$ and $\gamma>0$,
there   is an approximation-preserving reduction from $\uhTutte(q,\gamma)$ 
to the following problem.

\begin{description}
\item[Problem] $\TwoWeightFerroTutte(q)$.
\item[Instance] Graph $\graph=(\graphvertices,\graphedges)$
with an edge-weight function $\boldgamma': \graphedges \rightarrow
\{\gamma',\gamma''\}$
where $\gamma'$ and $\gamma''$ are rationals in the interval $[|\graphvertices|^{-3},1]$.
\item[Output]  $\ZTutte(G;q,\boldgamma')$.
\end{description}

Finally, Section~\ref{shift} gives an approximation-preserving reduction
from 
the problem
$\TwoWeightFerroTutte(q)$ to 
the problem $\Tutte(q,\gamma)$ for any  
$q>2$ and $\gamma>0$.

\section{Approximately computing the Tutte polynomial of a uniform hypergraph}
\label{bistohypertutte}

We first 
define a parameterised version of $\BIS$, and 
also a restricted version of this in which vertices on the right-hand side
are required to have the same degree.

\begin{description}
\item[Problem] $\BIS(\mu)$.
\item[Instance]  Bipartite graph $B$.
\item[Output]  $Z_{\IS}(B;\mu)=\sum_I\mu^{|I|}$, 
where the sum is over all independent sets $I$ in $B$.
\end{description}
 
\begin{description}
\item[Problem] $\srBIS(\mu)$.
\item[Instance]  Bipartite graph $B=(U,V,E)$ in which every vertex in $V$ has
the same degree.
\item[Output]  $Z_{\IS}(B;\mu)=\sum_I\mu^{|I|}$, 
where the sum is over all independent sets $I$ in $B$.
\end{description}

\begin{lemma}
\label{lem:bis}
Suppose that
$\mu>0$ is efficiently approximable. Then $\BIS \APred \srBIS(\mu)$.
\end{lemma}

\begin{proof} 
Lemma~15 of~\cite{APred}  gave 
an AP-reduction from 
$\BIS$ to \BmaxIS, the problem of counting \emph{maximum}
independent sets in a bipartite graph.
Here,  we give two AP-reductions --- first, a reduction from \BmaxIS\ to
$\BIS(\mu)$,
and then a reduction from 
the intermediate problem
$\BIS(\mu)$ to $\srBIS(\mu)$.

Let $B$ be an instance of $\BmaxIS$ with $n$ vertices and $m$ edges
and let $\varepsilon$ be the desired accuracy of the
approximation-preserving reduction.
Let $\xi$ be the size of a maximum independent set of~$B$
and let $Y$ be the number of maximum independent sets.
To do the construction, we first need a rough estimate of $\mu$,
so compute a rational value $\tilde{\mu}$ in the range
$\tfrac34 \mu \leq \tilde{\mu} \leq \tfrac54 \mu$.
Let 
$s$ be an integer satisfying
$$ s- 1 \leq \left\lceil \frac{n+3}{\lg(1+2\tilde \mu/5
)}\right\rceil \leq s.$$
Note that $s \geq \frac{n+3}{\lg(1+\mu/2)}$.
Let $B'$ be the graph with 
vertex set $\{(u,i)\mid u\in V(B), i\in [s]\}$
and edge set $\{((u,i),(v,j)) \mid (u,v)\in E(B)\}$.
Next, we need a more accurate estimate of~$\mu$.
Compute a rational number $\hat \mu$
in the range
$e^{-\epsilon/(60 ns)} \mu \leq \hat \mu \leq e^{\epsilon/(60
  ns)}\mu$.
Now since there are $n s$ vertices in $B'$,
$$e^{-\epsilon/60} Z_{\IS}(B';\mu) \leq Z_{\IS}(B';\hat \mu) \leq e^{\epsilon/60}
Z_{\IS}(B';\mu).$$ Thus, by using our oracle for $\BIS(\mu)$ with accuracy parameter $\epsilon/60$,
we can compute a value $Z$ satisfying
\begin{equation}
\label{APredeq}
e^{-\epsilon/30} Z_{\IS}(B';\hat \mu) \leq  Z \leq e^{\epsilon/30}
Z_{\IS}(B';\hat \mu).
\end{equation}
Now note that every independent set of~$B'$ points out an independent
set of~$B$. The vertex~$u$ of~$B$ is in the independent set of~$B$ if
there is at least one vertex~$(u,i)$ that is in the independent set
of~$B'$.
A size-$k$ independent set of~$B$ thus makes a contribution of
${({(1+\hat\mu)}^s-1)}^k$ to $Z_{\IS}(B';\hat \mu)$.
Since the number of independent sets of~$B$ is at most~$2^n$,
we have
\begin{equation}
\label{APredanother}
Y \leq \frac{Z_{\IS}(B';\hat \mu)}{{({(1+\hat\mu)}^s-1)}^\xi} \leq
Y + \frac{2^n}{({(1+\hat\mu)}^s-1)}
\leq Y + \frac14,
\end{equation}
where the final inequality follows from the definition of~$s$ since
$\hat \mu \geq \mu/2$.
The following simple procedure now gives a sufficiently accurate
estimate  for~$Y$. Take the value~$Z$ from Equation~(\ref{APredeq}),
divide it by ${({(1+\hat\mu)}^s-1)}^\xi$ and round down to the nearest
integer.
The fact that the accuracy is sufficient follows from~(\ref{APredeq})
and~(\ref{APredanother}). See~\cite[Theorem~3]{APred}.

Finally, we present an AP-reduction from  $\BIS(\mu)$ to $\srBIS(\mu)$.
Let $B=(U,V,E)$ be an $n$-vertex instance of $\BIS(\mu)$ in which the
maximum degree of a vertex in $V$ is $d>1$.
We will construct $B'=(U',V',E')$ -- an instance of $\srBIS(\mu)$
in which every vertex in $V'$ has degree~$d$.
Let $\varepsilon$ be the desired accuracy of the
approximation-preserving reduction.
As above, compute a rational value $\tilde{\mu}$
such that 
$\tfrac45 \tilde \mu \leq \mu \leq \tfrac43 \tilde \mu$.
Let $\mu^-$ denote the computed lower bound
$\tfrac45 \tilde \mu$ and let $\mu^+$ denote the computed upper bound
$\tfrac43 \tilde \mu$.
Let $D(x) = {(1+x)}^{d-1}$, $U(x) = x D(x)$,  
$L(s,x) = (1+x)^s$, and $Y(s,x) = L(s,x) + D(x)-1$. 
We start by computing an integer~$s$ such that
\begin{equation}
\label{eq:s} 
\frac{1}{L(s,\mu^-)}
\max(D(\mu^+)-1,U(\mu^+))
\leq \frac{\varepsilon}{6 dn}. 
\end{equation}
Note that $s=O( n \epsilon^{-1})$
so we can efficiently compute such a value~$s$ by starting with $s=1$,
and increasing $s$ one-by-one until we find a value for which~(\ref{eq:s}) holds.

Now
let $\Psi$ be a complete bipartite graph
with vertex sets $\{z_1,\ldots,z_{d}\}$ and $\{y_1,\ldots,y_s\}$.
For each vertex $v\in V$ of degree $\delta$,
take $d-\delta$ new copies of $\Psi$, and attach $v$ to vertex~$z_1$
of each copy.
Let $g < d n $ be the number of copies of $\Psi$ that get included in $B'$.

Now
$   Z_{\IS}(\Psi;\mu) = L(s,\mu) + (1+\mu)D(\mu)-1$
and the total contribution to $Z_{\IS}(\Psi;\mu)$ from independent
sets including the vertex~$z_1$ is $U(\mu)$.
Let $Y = Z_{\IS}(\Psi;\mu) -  U(\mu) = Y(s,\mu)$.
Clearly, $Y$ is the total contribution to $Z_{\IS}(\Psi;\mu)$ from independent
sets not including the vertex~$z_1$.

Then 
$$Z_{\IS}(B;\mu) Y^g \leq Z_{\IS}(B';\mu) \leq Z_{\IS}(B;\mu)
Z_{\IS}(\Psi;\mu)^g$$
so
$$Z_{\IS}(B;\mu)   \leq \frac{Z_{\IS}(B';\mu)}{Y^g} \leq Z_{\IS}(B;\mu)
{\left(\frac{Z_{\IS}(\Psi;\mu)}{Y}\right)}^g$$

Now 
\begin{itemize}
\item $Z_{\IS}(B';\mu)$ may be estimated with accuracy
parameter~$\epsilon/3$ using the
oracle.
\item $Y^g$ may be estimated directly (to the same accuracy) 
by computing a value~$\hat \mu$
satisfying
$e^{-\epsilon/(6 d n s)} \mu \leq \hat \mu \leq e^{\epsilon/(6 d n s)} \mu $
which ensures that
$$e^{-\epsilon/(6 d n )} L(s,\mu) \leq L(s,\hat \mu) \leq e^{\epsilon/(6 d n ) }L(s,\mu) $$
And noting from (\ref{eq:s}) that
$$L(s,\mu) \leq Y = L(s,\mu)\left(1+\frac{D(\mu)-1}{L(s,\mu)}\right) \leq L(s,\mu) e^{\epsilon/(6 dn)}.$$
\item Finally, (\ref{eq:s}) gives
$${\left(\frac{Z_{\IS}(\Psi;\mu)}{Y }\right)}^g = {\left(1+ \frac{U(\mu)}{Y(s,\mu)}\right)}^g \leq e^{\epsilon/3},$$
which finishes the reduction.
\end{itemize}
 \end{proof}

We are interested in the situation 
$\mu>1$ so that we have $q=\mu+1>2$   in 
the parameters of $\uhTutte(q,\gamma)$ in
Lemma~\ref{lem:one} below.

\begin{lemma}
\label{lem:one}
Suppose that
$\mu>0$ is efficiently approximable. Then
$$\srBIS(\mu) \APred\uhTutte(\mu+1,\mu).$$
\end{lemma}

\begin{proof}
Let $B=(U,V,E)$ be an instance of $\srBIS(\mu)$.  
Let  $\hypergraph=(\hypervertices,\hyperedges)$ be an instance of $\uhTutte(\mu+1,\mu)$
constructed as follows.
Let~$s$ be a new vertex that is not in $U\cup V$ and let $\hypervertices=U\cup\{s\}$. 
For $v\in V$, let $\Gamma(v)=\{u\in U:(u,v)\in E\}$
and let
$F_v=\Gamma(v) \cup\{s\}$.
Let
$\hyperedges=\bigcup_{v\in V} F_v$.
We will show below that 
\begin{equation}
\label{eq:temp}
Z_{\IS}(B,\mu)=(\mu+1)^{-1}\ZTutte(\hypergraph;\mu+1,\mu).
\end{equation}

Now let 
$\varepsilon$ be the desired accuracy in the approximation-preserving
reduction. To complete the reduction,
we first compute a value~$x$
in the range
$$e^{-\epsilon/2}
(\mu+1)^{-1} \leq x \leq 
e^{\epsilon/2}
(\mu+1)^{-1}.
$$
This is easy to do since $\mu$ is efficiently approximable
and  
$$\frac{e^{-\epsilon/2}}{\mu+1}
\leq \frac{1}{e^{\epsilon/2}\mu+1}\mbox{ and }
\frac{1}{e^{-\epsilon/2}\mu+1} \leq
\frac{e^{\epsilon/2}}{\mu+1}.$$
Then we use the oracle to estimate
$\ZTutte(\hypergraph;\mu+1,\mu)$ with accuracy parameter~$\varepsilon/2$.

We finish the proof by establishing~(\ref{eq:temp}).
Let $S\subseteq V$.  The contribution to $Z_{\IS}(B,\mu)$ 
from independent sets $I$ with $I\cap V=S$ is 
$\mu^{|S|}(\mu+1)^{|U|-|\Gamma(S)|}$, where $\Gamma(S)=\bigcup_{v\in S}\Gamma(v)$.
That is,
$$
Z_{\IS}(B,\mu)=\sum_{S\subseteq V}\mu^{|S|}(\mu+1)^{|U|-|\Gamma(S)|}.
$$

On the other hand, the contribution to $\ZTutte(\hypergraph;\mu+1,\mu)$
from the hyperedge set $\hyperedges=\{F_v:v\in S\}$ is 
$$\mu^{|S|}(\mu+1)^{\kappa(\hypervertices,\hyperedges)}=
\mu^{|S|}(\mu+1)^{|U|-|\Gamma(S)|+1},$$  
since the vertices in $\Gamma(S)$ together
with $s$ form one connected component, and all other vertices are isolated.
Thus
$$\ZTutte(\hypergraph;\mu+1,\mu)=
\sum_{S\subseteq V}\mu^{|S|}(\mu+1)^{|U|-|\Gamma(S)|+1}.$$
\end{proof}

\section{Approximately computing the multivariate Tutte polynomial of a graph}
\label{sec:tographs}
 
We have already given a reduction from $\BIS$ to 
$\uhTutte(\mu+1,\mu)$ for any efficiently approximable $\mu>0$.
In this section, we prove the following lemma, reducing the
latter to $\TwoWeightFerroTutte$.
 
\begin{lemma}
Suppose that
$q>2$ and $\gamma>0$ are efficiently approximable. Then
$\uhTutte(q,\gamma) \APred \TwoWeightFerroTutte(q)$.
\label{lem:last}
\end{lemma}

\begin{proof}
Start with a $\terminals$-uniform hypergraph
$\hypergraph=(\hypervertices,\hyperedges)$.
$\hypergraph$ is an instance of
the problem
$\uhTutte(q,\gamma)$.  
For convenience, let $\hypervertices=\{v_1,\ldots, v_\numvert\}$ 
and $\hyperedges=\{\hyperedge_1,\ldots, \hyperedge_m\}$.

The basic idea is  to simulate each hyperedge with a copy of the
gadget $\gadget$ from Section~\ref{sec:gadget}
  (actually with the variant $\vargadget$ from Section~\ref{sec:variant}).
So, from~$\hypergraph$, we construct a graph $\hatG$ which contains the vertices of~$\hypergraph$
and also, for each hyperedge~$\hyperedge_j$ of~$\hypergraph$, a clique $K_j$.
This clique, together with the vertices in~$\hyperedge_j$, forms a copy of the gadget~$\vargadget$
in which vertices in $\hyperedge_j$ are the terminals. Given the phase transition from Section~\ref{sec:gadget},
the only likely configurations of the gadget (which are the configurations
that contribute substantially to the Tutte polynomial of~$\hatG$)
are (1) those that put all of the vertices in~$f_j$ in the same connected component, and
(2) those that put the vertices in~$f_j$ in distinct connected components.
Thus, $\hatG$ will simulate the hypergraph~$\hypergraph$ (for the purpose of approximately computing 
Tutte polynomials).
  
We now describe the details of the construction.  
We start by setting the parameters. Let $\varepsilon$ be the desired accuracy in the approximation-preserving
reduction (see the definition of an AP-reduction in Section~\ref{sec:APred}).
Let $\chi = \varepsilon/(4m)$ 
The tolerance $\tol$, which is a parameter to Lemma~\ref{lem:excludedmiddletwo}, 
may be chosen to be any value such that
$$
\left(1+ \frac{\tol(1+e^{\chi}\gamma)}{1-\tol}\right)
\leq e^\chi.$$
Note that this inequality is achieved for $\eta=O(\varepsilon/m)$.

The reduction will construct an instance~$\Ghat=(\graphvertices,\graphedges)$ of 
the target problem
$\TwoWeightFerroTutte(q)$
along with an appropriate edge-weight function $\boldgamma': \graphedges \rightarrow
\{\gamma',\gamma''\}$. We will show that using
an oracle to approximate the solution to this instance with accuracy
parameter $\delta=\varepsilon/2$ enables us to get within
$\exp(\pm\varepsilon)$ of
$\ZTutte(H;q,\gamma)$.

Let $\lambda = \lambda_c + (q-\lambda_c)/2$ and let
$\cliquesize_0$ be the quantity from 
Lemmas~\ref{lem:excludedmiddleone} and~\ref{lem:excludedmiddletwo}. 
Let $\cliquesize$ be the smallest integer greater than 
$\max\{\terminals^{16},\tol^{-1/8},\cliquesize_0\}$
for which $\cliquesize^{1/4}$ is an integer.
Using the algorithm from Lemma~\ref{lem:computerho}
(suitably powered up so that its failure probability is at most~$3/8$)
we can compute a rational ${\critprob}$ in the
range $[\cliquesize^{-3},\lambda/\cliquesize]$
such that, if  
$\edgesubset$ is drawn from $\RC(\gadget;q,p)$, then
Equation~(\ref{eq:compbalance}) holds.
Note from  Lemma~\ref{lem:excludedmiddletwo} that 
Equation~(\ref{eq:dichotomy}) holds.

The construction is as follows. For every $j\in[m]$, let $\clique_j$ be a set of $\cliquesize$ vertices
and let $\graphvertices_j = \clique_j \cup \hyperedge_j$.
Let 
$\graphedges_j = \clique_j^{(2)} \cup (\clique_j \times \hyperedge_j)$ 
and let $G_j$ be the graph $(\graphvertices_j,\graphedges_j)$.
Note that $\graph_j$ is a copy of the graph $\vargadget$ from Section~\ref{sec:variant} --- $\hyperedge_j$
is the set of terminals.
Now we construct the graph $\Ghat = (\graphvertices,\graphedges)$ 
where $\graphvertices = \hypervertices \cup \bigcup_{j\in[m]} 
\clique_j$ 
and $\graphedges = \bigcup_{j\in[m]} \graphedges_j$.
Define 
 $$ \hat\edgeprob(e) = \begin{cases}
\critprob, &\text{if $e\in \clique_j^{(2)}$ for some~$j\in[m]$, and}\\
\cliquesize^{-3/4}, &\text{if $e\in \clique_j\times \hyperedge_j$ for some~$j\in[m]$.} 
\end{cases}
$$ 
Let $\gamma_e = \hat\edgeprob(e)/(1-\hat\edgeprob(e))$, and let $\boldgamma'=\{\gamma_e\}_{e\in\graphedges}$.
Note that $\gamma' = \critprob/(1-\critprob)$ 
and $\gamma'' = 
\cliquesize^{-3/4}/(1-\cliquesize^{-3/4})$. These are both rational values in
the interval $[|\graphvertices|^{-3},1]$, as required, since $\lambda/\cliquesize \leq 1/2$
by Lemma~~\ref{lem:excludedmiddleone}.

Denote by $\Pi_\edge$ the set of all partitions of $\hyperedge_\edge$.
Given  $\pi_\edge\in\Pi_\edge$,
let $\calA_\edge^{\pi_\edge}$ 
be the collection of all edge subsets $\edgesubset_\edge \subseteq\graphedges_\edge$
that induce the partition $\pi_\edge$ (into connected components) on $\hyperedge_\edge$.
For $\edgesubset_\edge \subseteq \graphedges_\edge$, let $\rv(\edgesubset_\edge)$ denote 
the number of 
connected components that contain 
terminals (vertices in $\hyperedge_\edge$) in the graph 
$(\graphvertices_\edge,\edgesubset_\edge)$.
Let $\kappa'(\graphvertices_\edge,\edgesubset_\edge) = \kappa(\graphvertices_\edge,\edgesubset_\edge) - 
\rv(\edgesubset_\edge)$ be the number of remaining connected components (that do not contain terminals).
For partitions $\pi$ and $\pi'$, $\pi\vee\pi'$ denotes the
finest partition that is a common coarsening of $\pi$ and
$\pi'$. 
If two elements are 
in the same block
in $\pi$ or in $\pi'$ then they are 
in the same block
in the coarsening.
The ``coarsest'' partition has one block and the ``finest'' partition consists of singleton blocks.
Technically, $\pi_\edge\in \Pi_\edge$ is a partition of $\hyperedge_\edge$
but we consider it as a partition of the entire vertex set $\hypervertices$
by extending it with singleton blocks.  

Now $$\ZTutte(\Ghat;q,\boldgamma')=\sum_{\edgesubset \subseteq \graphedges}
q^{\kappa(\graphvertices,\edgesubset)}
\prod_{e\in\graphedges}
\gamma_e.$$
Let $\edgesubset_\edge = \edgesubset \cap \graphedges_\edge$.
Let $\gamma({\edgesubset_\edge}) = \prod_{e\in \edgesubset_\edge} \gamma_e$.
Then  
\begin{align*}\ZTutte(\Ghat;q,\boldgamma')&=
\sum_{\textstyle{\pi_\edge \in\Pi_\edge \atop \forall \edge \in[m]}}\>
\sum_{\textstyle{\edgesubset_\edge\in \calA_j^{\pi_j} 
  \atop\forall \edge\in[m]}}\>
q^{\kappa(\graphvertices,
\bigcup_{\edge\in[m]} \edgesubset_\edge)}
\prod_{\edge\in[m]} 
\gamma(\edgesubset_\edge)\\
&=
\sum_{\textstyle{\pi_\edge \in\Pi_\edge \atop \forall \edge \in[m]}}\>
q^{\kappa(\pi_1\vee\cdots\vee\pi_m)}
\sum_{\textstyle{\edgesubset_\edge\in \calA_j^{\pi_j} \atop\forall \edge
  \in[m]}}\>
\prod_{\edge\in[m]} 
\gamma(\edgesubset_\edge)
q^{\kappa'(\graphvertices_\edge,
  \edgesubset_\edge)},
\end{align*}
where $\kappa(\pi_1\vee\cdots\vee\pi_m)$ denotes the number of blocks
in the
partition $\pi_1\vee\cdots\vee\pi_m$.

Now let
$$
Z_\edge(\pi_\edge)=\sum_{\edgesubset_\edge\in\calA_\edge^{\pi_\edge}}\gamma(\edgesubset_\edge)\,q^{\kappa'(\graphvertices_\edge
  ,\edgesubset_\edge)}
$$
and pull out the contribution of each~$\edge$ to get
\begin{equation}\label{eq:ZPottsbasic}
\ZTutte(\Ghat;q,\boldgamma')
=
\sum_{\textstyle{\pi_\edge \in\Pi_\edge \atop \forall \edge \in[m]}}\>
q^{\kappa(\pi_1\vee\cdots\vee\pi_m)} 
\prod_{\edge=1}^m Z_\edge(\pi_\edge).
\end{equation}

At this point, we make some connections to Section~\ref{sec:variant}.
Let $\Pi_\edge^k$ denote the set of all partitions of
$\hyperedge_\edge$ into $k$ blocks and let
$Z_\edge^k = \sum_{\pi_\edge\in\Pi_\edge^k} Z_\edge(\pi_\edge)$.
From Lemma~\ref{lem:wiring}
we deduce that
\begin{equation}
\label{eq:usewiring}
\frac{Z_\edge^k}{Z_\edge^{k'}} = 
\frac{\Pr(\rv(\edgesubset)=k)}{\Pr(\rv(\edgesubset)=k')},
\end{equation}
where $\edgesubset$ is drawn from $\RC(\gadget;q,\edgeprob)$
for the function~$\edgeprob$ from Section~\ref{sec:gadget}.

Denote by $\bot_\edge$ the finest partition (with $\terminals$ blocks) of
$\hyperedge_\edge$
and by $\top_\edge$ the coarsest partition (with one block)
on the same set.  
By Lemma~\ref{lem:computerho} Equation~(\ref{eq:compbalance}) 
and Equation~(\ref{eq:usewiring})
we have 
\begin{equation}
\label{chione}
e^{-\chi} \gamma \leq
\frac
{Z_\edge(\top_\edge)} 
{Z_\edge(\bot_\edge)} \leq e^{\chi} \gamma.
\end{equation}
 
Use the FPRAS from Lemma~\ref{lem:computeZ}
(again, suitably powered up so that its failure probability is at most~$3/8$)
to compute a value~$\lastc$ 
which satisfies
\begin{equation}
\label{chitwo}
 e^{-\chi} Z_\edge(\bot_\edge) \leq \lastc \leq e^{\chi} Z_\edge(\bot_\edge).
\end{equation}

Now, define $b_\edge(0)=\bot_\edge$ and $b_\edge(1)=\top_\edge$.
Then, by restricting
the sum in  (\ref{eq:ZPottsbasic}) 
to terms satisfying $\pi_\edge\in\{\bot_\edge,\top_\edge\}$ for all $\edge$,
we get
$$\ZTutte(\Ghat;q,\boldgamma')\geq 
\sum_{z\in\{0,1\}^m} 
   q^{\kappa(b_1(z_1)\vee\cdots\vee b_m(z_m))}
\prod_{j=1}^m Z_j(b_j(z_j)).
$$
Letting $\|z\|$ denote the Hamming weight of $z$
and using Equation~(\ref{chione}),
the right-hand side is at least
\begin{align*}
&\sum_{z\in\{0,1\}^m} 
   q^{\kappa(b_1(z_1)\vee\cdots\vee b_m(z_m))}
{(e^{-\chi} \gamma)}^{\|z\|}   \prod_{j=1}^m Z_j(\bot_j)\\
&\qquad\geq
e^{-\chi m}\prod_{j=1}^m Z_j(\bot_j)
\sum_{z\in\{0,1\}^m}\notag
   q^{\kappa(b_1(z_1)\vee\cdots\vee b_m(z_m))} {\gamma}^{\|z\|}.
\end{align*}
Using Equation~(\ref{chitwo}), we get
\begin{align}
\ZTutte(\Ghat;q,\boldgamma')&\geq \label{eq:ineq1}
\lastc^m e^{-2\chi m}\sum_{z\in\{0,1\}^m}
   q^{\kappa(b_1(z_1)\vee\cdots\vee b_m(z_m))} {\gamma}^{\|z\|}\\
   &=\lastc^m e^{-2\chi m}\,\ZTutte(\hypergraph;q,\gamma).\notag
\end{align}

To finish, we just have to show that 
this lower bound for $\ZTutte(\Ghat;q,\boldgamma')$ is actually a good
estimate
because
the terms that we threw away
in inequality~(\ref{eq:ineq1}) don't amount to much.
For $\edge\in[m]$, let 
$\Pi_\edge^1=\{\top_\edge\}$ and let
$\Pi_\edge^0=\Pi_\edge\setminus\{\top_\edge\}$.

Then, starting from (\ref{eq:ZPottsbasic}),
$$\ZTutte(\Ghat;q,\boldgamma') 
= \sum_{z\in\{0,1\}^m}
\sum_{\textstyle{\pi_\edge \in\Pi_\edge^{z_\edge} \atop \forall \edge \in[m]}}\>
q^{\kappa(\pi_1\vee\cdots\vee\pi_m)} 
\prod_{\edge=1}^m Z_\edge(\pi_\edge).$$
Now note the partition $\pi_\edge \in\Pi_\edge^{z_\edge} $ is refined by
$b_\edge(z_\edge)$. Thus, 
$b_1(z_1)\vee\cdots\vee b_m(z_m)$ has
at least as many connected
components as $\pi_1\vee\cdots\vee\pi_m$.
So we get
\begin{align*}\ZTutte(\Ghat;q,\boldgamma') 
&\leq \sum_{z\in\{0,1\}^m}
\sum_{\textstyle{\pi_\edge \in\Pi_\edge^{z_\edge} \atop \forall \edge \in[m]}}\>
q^{\kappa(b_1(z_1)\vee\cdots\vee b_m(z_m))}
\prod_{\edge=1}^m Z_\edge(\pi_\edge)\\
 &=\sum_{z\in\{0,1\}^m}\>
   q^{\kappa(b_1(z_1)\vee\cdots\vee b_m(z_m))}
\sum_{\textstyle{\pi_\edge \in\Pi_\edge^{z_\edge} \atop \forall \edge \in[m]}}\>
      \prod_{\edge=1}^m Z_\edge(\pi_\edge)\notag\\
   &=\sum_{z\in\{0,1\}^m}\>
   q^{\kappa(b_1(z_1)\vee\cdots\vee b_m(z_m))}
   \prod_{\edge=1}^m Z_\edge(\Pi_\edge^{z_\edge}),\notag\\
\end{align*}
where $Z_\edge(\Pi_\edge^{z_\edge}) = 
\sum_{\pi_\edge \in\Pi_\edge^{z_\edge}}
      Z_\edge(\pi_\edge)$.
Now, $Z_\edge(\Pi_\edge^1)=Z_\edge(\top_\edge) \leq e^{2\chi} \gamma \lastc$.
Also, $Z_\edge(\Pi_\edge^0) = \sum_{k=2}^\terminals Z_\edge^k  $.
By Equation~(\ref{eq:usewiring}), this is
$$Z_j^\terminals\left(1 + \frac{\sum_{k=2}^{\terminals-1} \Pr(\rv(A)=k)}
{\Pr(\rv(A) = \terminals)}\right) ,$$
where $\edgesubset$ is drawn from $\RC(\gadget;q,\edgeprob)$
for the function~$\edgeprob$ from Section~\ref{sec:gadget}. 
Since $Z_j^\terminals\leq e^{\chi}\lastc$, we have
$$Z_\edge(\Pi_\edge^0) 
\leq e^{\chi} \lastc \left(1 + \frac{\sum_{k=2}^{\terminals-1}
    \Pr(\rv(\edgesubset)=k)}{\Pr(\rv(\edgesubset)=\terminals)}\right).$$
Now by Lemma~\ref{lem:excludedmiddletwo}
Equation~(\ref{eq:dichotomy}),
$\sum_{k=2}^{\terminals-1}
    \Pr(\rv(\edgesubset)=k)\leq \tol$.
Thus, $\Pr(\rv(\edgesubset)=1)+\Pr(\rv(\edgesubset)=\terminals)\geq 1-\tol$ and
by Equation~(\ref{eq:compbalance}),
$$\Pr(\rv(\edgesubset)=\terminals)\geq \frac{1-\tol}{1+e^{\chi}\gamma}.$$
Thus,
$$
Z_\edge(\Pi_\edge^0) 
\leq e^{\chi}\lastc 
\left(1+ \frac{\tol(1+e^{\chi}\gamma)}{1-\tol}\right)
\leq e^{2 \chi} \lastc  
 .$$ 
Thus, 
$$\ZTutte(\Ghat;q,\boldgamma') \leq 
 e^{2 \chi m} \lastc^m 
\sum_{z\in\{0,1\}^m}\>
   q^{\kappa(b_1(z_1)\vee\cdots\vee b_m(z_m))}
  \gamma^{\|z\|} = e^{2 \chi m} \lastc^m  \ZTutte(\hypergraph;q,\gamma)
.$$

Recall that $\chi = \varepsilon/(4m)$.
Thus, if we knew a quantity  $\psi$ in the
range 
$$e^{-\epsilon/2} \leq \frac{\psi}{\ZTutte(\Ghat;q,\boldgamma')}
\leq e^{\epsilon/2},$$
we would have
$$e^{-\epsilon} \leq \frac{\psi \lastc^{-m}}{\ZTutte(\hypergraph;q,\gamma) }
\leq e^{\epsilon},$$
 so we have completed the AP-reduction.
\end{proof}

 \section{Approximately computing the Tutte polynomial of a graph}
\label{shift}

In this section, we complete the approximation-preserving reduction
from~\BIS\ to $\Tutte(q,\gamma)$
by giving an approximation-preserving reduction from 
the problem
$\TwoWeightFerroTutte(q)$ to 
the problem
$\Tutte(q,\gamma)$.

For this, we first need to define what it means to ``implement'' an
edge weight. This description is mainly taken from \cite[Section~1.6]{planartutte}.
Fix $q>2$.
Let 
$W$ be a set of edge weights.
For example, $W$ might contain the edge weight
$\gamma$ (a parameter of the problem $\Tutte(q,\gamma)$) or
it might contain a selection of edge weights that we have already
implemented using~$\gamma$.
Let $\gamma^*$ be a weight (which may not be in $W$) which we want to ``implement''.
Suppose that 
there is a graph~$\Upsilon$,
with distinguished vertices $s$ and~$t$ 
and  an edge-weight function $\hat\boldgamma: \graphedges(\Upsilon) \rightarrow
W$ 
such that
\begin{equation}
\label{eq:implement}
\gamma^* = \frac{q Z_{st}(\Upsilon)}{Z_{s|t}(\Upsilon)},
\end{equation}
where $Z_{st}(\Upsilon)$ denotes the contribution to
$\ZTutte(\Upsilon;q,\hat\boldgamma)$ arising from edge-sets $A$ in which $s$ and $t$ are
in the same component.
That is, 
$Z_{st}(\Upsilon) = \sum_{A} \hat\boldgamma(A) q^{\kappa(V,A)}$, 
where the sum is over subsets $A\subseteq E(\Upsilon)$ in which 
$s$ and $t$ are in the same component.
Similarly, $Z_{s|t}$ denotes the contribution to 
$\ZTutte(\Upsilon;q,\hat\boldgamma)$ arising from edge-sets $A$ in which $s$ and $t$ are in different components.
In this case, we say that $\Upsilon$ and $\hat\boldgamma$ implement
$\gamma^*$
(or even that $W$ implements $\gamma^*$).

The purpose of ``implementing''  edge weights is this.
Let $G$ be a graph with edge-weight function $\boldgamma$.
Let $f$ be some edge of $G$ with edge weight $\gamma_f=\gamma^*$.
Suppose that $W$ implements $\gamma^*$.
Let $\Upsilon$ be a graph with distinguished vertices $s$ and $t$
with a weight function $\hat\boldgamma$ satisfying (\ref{eq:implement}). 
Construct the weighted graph $\widetilde G$
by replacing edge $f$ with a copy of $\Upsilon$ (identify $s$ with either endpoint of $f$
--- it doesn't matter which one --- and identify $t$ with the other endpoint of $f$ and remove edge $f$).
Define the weight function $\tilde\boldgamma$ as follows:
$$\tilde\gamma_e = \begin{cases}
\hat\gamma_e, &\text{if $e\in E(\Upsilon)$, and}\\
\gamma_e, & \text{otherwise.}
\end{cases}$$
Then the definition of the multivariate Tutte polynomial gives
the following equation, which we will justify below.
\begin{equation}
\label{eq:shift}
\ZTutte(\widetilde G;q,\tilde\boldgamma) = \frac{Z_{s|t}(\Upsilon)}{q^2} \ZTutte(G;q,\boldgamma).\end{equation}
So, as long as  $Z_{s|t}(\Upsilon)$ is easy to evaluate
and $q$ is efficiently approximable,
approximating the multivariate Tutte polynomial of $\widetilde G$ with weight function $\tilde\boldgamma$ is
essentially the same as approximating the multivariate Tutte polynomial of $G$ with weight function~$\boldgamma$.

Let us now justify Equation~(\ref{eq:shift}).
Let $G=(V,E)$.
Let $\mathcal{A}_{st}$ be the set of all subsets $A\subseteq E-\{f\}$ in which 
vertices~$s$ and~$t$ are in the same component.
Let $\mathcal{A}_{s|t}$ be the set of all other subsets~$A$ of~$E-\{f\}$.
Let $\gamma_A = \prod_{e \in A} \gamma_e$.
Using the definition of $\gamma^*$ from (\ref{eq:implement}), $\ZTutte(G;q,\boldgamma)$ can be written
as
$$\ZTutte(G;q,\boldgamma)= \sum_{A \in \mathcal{A}_{st}} \gamma_A q^{\kappa(V,A)}\left(1 + \frac{q Z_{st}(\Upsilon)}{Z_{s|t}(\Upsilon)}\right)
+ \sum_{A \in \mathcal{A}_{s|t}} \gamma_A q^{\kappa(V,A)} \left(1 + \frac{Z_{st}(\Upsilon)}{Z_{s|t}(\Upsilon)}\right).
$$
Multiplying through by $Z_{s|t}(\Upsilon)/q^2$,
the right-hand-side of~(\ref{eq:shift}) is the following sum.
\begin{align*}
&\sum_{A \in \mathcal{A}_{st}} \gamma_A q^{\kappa(V,A)-2} Z_{s|t}(\Upsilon) +
\sum_{A \in \mathcal{A}_{st}} \gamma_A q^{\kappa(V,A)-1} Z_{st}(\Upsilon) \\
&\qquad+\sum_{A \in \mathcal{A}_{s|t}} \gamma_A q^{\kappa(V,A)-2} Z_{s|t}(\Upsilon) +
\sum_{A \in \mathcal{A}_{s|t}} \gamma_A q^{\kappa(V,A)-2} Z_{st}(\Upsilon)
\end{align*}
Now let $\mathcal{A}'_{st}$ be the set of all subsets $A'\subseteq E(\Upsilon)$ in which 
vertices~$s$ and~$t$ are in the same component.
Let $\mathcal{A}'_{s|t}$ be the set of all other subsets~$A'$ of~$E(\Upsilon)$.
Paying attention to the number of connected components, we can re-write each
of the four terms to re-write the sum as
\begin{align*}
&\sum_{A \in \mathcal{A}_{st}} \sum_{A'\in \mathcal{A}'_{s|t}} \gamma_{A\cup A'} q^{\kappa(V,A\cup A')} +
\sum_{A \in \mathcal{A}_{st}} \sum_{A' \in \mathcal{A}'_{st}} \gamma_{A\cup A'} q^{\kappa(V,A\cup A')}  \\
&\qquad+\sum_{A \in \mathcal{A}_{s|t}} \sum_{A'\in \mathcal{A}'_{s|t}} \gamma_{A\cup A'} q^{\kappa(V,A\cup A')} +
\sum_{A \in \mathcal{A}_{s|t}} \sum_{A' \in \mathcal{A}'_{st}} \gamma_{A\cup A'} q^{\kappa(V,A\cup A')},
\end{align*}
which is the left-hand-side of (\ref{eq:shift}).

Two especially useful implementations 
(see, for example, \cite{JVW90})
are series and parallel compositions.
These are explained in detail in \cite[Section 2.3]{JacksonSokal}.
So we will be brief here.
Parallel composition is the case in which $\Upsilon$ consists of two parallel edges $e_1$ and $e_2$
with endpoints $s$ and $t$ and  $\hat\gamma_{e_1}=\gamma_1$ and $\hat\gamma_{e_2}=\gamma_2$.
It is easily checked from Equation~(\ref{eq:implement})
that $\gamma^* = (1+\gamma_1)(1+\gamma_2)-1$. Also, the extra factor in Equation~(\ref{eq:shift}) cancels,
so in this case $\ZTutte(\widetilde G;q,\tilde\boldgamma) = \ZTutte(G;q,\boldgamma)$.

Series composition is the case in which $\Upsilon$ is a length-2 path from $s$ to $t$ consisting of edges $e_1$ and $e_2$
with $\hat\gamma_{e_1}=\gamma_1$ and $\hat\gamma_{e_2}=\gamma_2$.
It is easily checked from Equation~(\ref{eq:implement})
that $\gamma^* =  \gamma_1\gamma_2/(q+\gamma_1+\gamma_2)$. 
Also, the extra factor in Equation~(\ref{eq:shift}) is $q+\gamma_1+\gamma_2$,
so in this case $\ZTutte(\widetilde G;q,\tilde\boldgamma) = (q+\gamma_1+\gamma_2) \ZTutte(G;q,\boldgamma)$.
It is helpful to note that
$\gamma^*$ satisfies
$$\left(1+\frac{q}{\gamma^*}\right) = \left(1+\frac{q}{\gamma_1}\right) \left(1+\frac{q}{\gamma_2}\right).$$
We are now ready to prove this lemma.

\begin{lemma}
\label{lem:shiftaround}
Suppose that
$q>2$ and $\gamma>0$ are efficiently approximable. Then
$\TwoWeightFerroTutte(q) \APred \Tutte(q,\gamma)$.
\end{lemma}

\begin{proof}
Let  $\graph=(\graphvertices,\graphedges)$ be an instance of
$\TwoWeightFerroTutte(q)$
with edge-weight function $\boldgamma': \graphedges \rightarrow
\{\gamma',\gamma''\}$ 
where $\gamma'$ and $\gamma''$ are rationals in the interval $[|\graphvertices|^{-3},1]$.
We will assume without loss of generality that $|E|$ is sufficiently
large with respect to the fixed parameters~$q$ and~$\gamma$.  

Let $\epsilon$ be the desired accuracy in the approximation-preserving reduction.
Let 
$\chi = \epsilon/(4(|\graphvertices|+|\graphedges|^2))$. 
Let $\hat q$ be a rational in the range
$e^{-\chi} q \leq \hat q \leq e^{\chi} q$ and
let $\hat\gamma$ be a rational in the range
$e^{-\chi} \gamma \leq \hat \gamma \leq e^{\chi} \gamma$.
Since $q$ and $\gamma$ are efficiently approximable,  the amount of time
that it takes to compute $\hat q$ and $\hat \gamma$
is at most a polynomial in $|\graphvertices|$, $|\graphedges|$ and $\varepsilon^{-1}$.

The idea of the proof is to show how to use 
series and parallel compositions from
the set $W=\{\hat \gamma\}$ to
implement edge-weights ${\gamma'}^*$ 
and ${\gamma''}^*$
satisfying
\begin{equation}
\label{eq:finalfinal}
e^{- \chi} {\gamma'} \leq  {\gamma'}^* \leq 
e^{ \chi} \gamma'
\end{equation}
and 
$$e^{- \chi} {\gamma''} \leq  {\gamma''}^* \leq 
e^{ \chi} \gamma''.$$
Letting $\boldgamma^*$ be the edge-weight function derived from~$\boldgamma'$ by
replacing $\gamma'$ with ${\gamma'}^*$ and ${\gamma''}$ with ${\gamma''}^*$,
Equation~(\ref{approxparams})  will then give
$$
e^{-\varepsilon /4}
\ZTutte(\graph;q,\boldgamma') \leq
\ZTutte(\graph;\hat q,\boldgamma^*) \leq 
e^{\varepsilon /4}
\ZTutte(\graph; q, \boldgamma').
$$

Let $\hat\boldgamma$ be the edge-weight function which assigns every edge weight $\hat\gamma$.
We can think of our implementations as constructing a graph $\widehat{G}$
such that 
$\ZTutte(\graph;\hat q,\boldgamma^*)$ is equal to 
the product of $\ZTutte(\widehat{G};\hat q,\hat \boldgamma)$
and an easily-computed function of~$\hat q$ and~$\hat \gamma$.\footnote{This easily-computed
function arises from  the extra factor $Z_{st}(\Upsilon)$ in Equation~(\ref{eq:implement}). It is easy to compute
because our implementations use only series and parallel composition}
We will ensure that each implementation uses at most $|\graphedges|$ edges,
so the total number of edges in $\widehat{G}$ is at most $|\graphedges|^2$.
To finish, we note (from (\ref{approxparams})) that 
$$e^{-\varepsilon/4 }
\ZTutte(\widehat{G};q,\boldgamma) \leq
\ZTutte(\widehat{G};\hat q,\hat \boldgamma) \leq 
e^{\varepsilon /4}
\ZTutte(\widehat{G}; q, \boldgamma),$$
where $\boldgamma$ is the constant edge-weight function which assigns every edge weight $\gamma$. 
We finish the approximation of 
$ \ZTutte(\widehat{G};\hat q,\hat \boldgamma)$
by using the oracle to approximate
$\ZTutte(\widehat{G};q,\boldgamma) $ using accuracy parameter $\delta = \epsilon/2$.

It remains to show how to do the implementations.
Taking 
$$\pi = \frac{\chi}{2|\graphvertices|^3} \leq \frac{\gamma' \chi}{2} \leq \gamma'(1-e^{-\chi}),$$
we show how to use $W=\{\hat \gamma\}$
to implement an edge-weight 
 ${\gamma'}^*$
which satisfies ${\gamma'} - \pi \leq {\gamma'}^* \leq {\gamma'}$.
This ensures that Equation~(\ref{eq:finalfinal}) holds. 
The implementation of ${\gamma''}^*$ is similar.

Our implementation  is taken from Section~2.1 of our
paper \cite{planartutte}. First, we can 
implement a weight $\gamma_1\leq\tfrac14$ 
by taking a series composition of $k$ edges of weight~$\hat\gamma$ 
for sufficiently large~$k$. It suffices to take
$$k = \left\lceil
\frac
{\log(1+4\hat q)}
{\log(1+\hat q/\hat \gamma)}
\right\rceil.$$
Then implement a weight $\gamma_j$ by taking a series composition of
$j$
copies of $\gamma_1$.
The following (recursive) definitions are
from \cite[Section~2.1]{planartutte} for integers~$j\geq 1$:

$$d_j = \left\lfloor
\frac{\log((1+{\gamma'})
\prod_{\ell=1}^{j-1}{(1+\gamma_\ell)}^{-d_\ell}
)
}
{\log(1+\gamma_j)}
\right\rfloor, \mbox{ and }
m = \left\lceil 
\frac
{\log( \hat q (1+{\gamma'}) /\pi+1)}{\log(
 \hat q/\gamma_1+1)}\right\rceil.
$$

Then the implementation combines, in parallel, $d_j$ edges with
edge-weight $\gamma_j$, for all $j\in[m]$.

The calculation in \cite[Section~2.1]{planartutte}
shows that the implemented value~${\gamma'}^*$ 
satisfies ${\gamma'} - \pi \leq {\gamma'}^* \leq {\gamma'}$, as required.

Now to finish we need to show that $d_1 + \cdots + d_m \leq |\graphedges|$, which we
used above.
First, note that the fixed parameters~$q$ and~$\gamma$
give fixed upper and lower bounds on
$\hat q$ and $\hat \gamma$.
Using the upper bound $\gamma'\leq 1$,  we see that
the value $m$ is
at most  
logarithmic in $\pi^{-1}$ which is at most logarithmic in $|\graphvertices|$,
$|\graphedges|$, and
$\varepsilon^{-1}$.
The calculation in~\cite[Section~2.1]{planartutte} shows that the same
is true of $d_1,\ldots,d_m$. In fact, there is a fixed upper bound
for~$d_j$ depending only on~$\hat q$ and~$\hat \gamma$. The proof of this fact uses the
fact that $0< \gamma_j \leq \gamma_1 \leq \tfrac14$.  
Since we assumed, without loss of generality,
that $|\graphedges|$ is sufficiently large with respect to the fixed parameters~$q$ and~$\gamma$,
we conclude that $d_1 + \cdots + d_m \leq |\graphedges|$, as required.
\end{proof}

\begin{proof}[of Theorem~\ref{thm:main}]
Theorem~\ref{thm:main} follows from Lemmas \ref{lem:bis}, \ref{lem:one},
\ref{lem:last} and~\ref{lem:shiftaround}.
\end{proof}

\section{$3$-Uniform Hypergraphs}
\label{sec:3uniform} 

Lemmas~\ref{lem:bis}
and \ref{lem:one} have the following corollary.

 \begin{corollary}
\label{cor:uhhard}
Suppose that
$q>0$ is efficiently approximable. Then
$\BIS \APred\uhTutte(q,q-1).$
\end{corollary}

Thus, assuming that there is no FPRAS for \BIS, we can conclude that there is no
FPRAS for computing the Tutte polynomial 
of a uniform hypergraph when the edge-weights are set to  $\gamma=q-1$.
The \emph{Ising model} corresponds to the $q=2$ case of the Potts model.
Thus, we conclude that
there is no FPRAS for computing the partition function of
the Ising model on a uniform hypergraph in which every edge has weight~$1$.
 
We conclude this paper with a contrasting positive result for $3$-uniform hypergraphs.
Consider the following problem.
\begin{description}
\item[Problem] $\u3hTutte(q,\gamma)$.
\item[Instance]  A $3$-uniform hypergraph $\hypergraph=(\hypervertices,\hyperedges)$.
\item[Output]  $\ZTutte(\hypergraph;q,\gamma)$,
where $\boldgamma$ is the constant function  with $\boldgamma_\hyperedge = \gamma$ for every $\hyperedge\in\hyperedges$.
\end{description}

\begin{lemma}
Suppose that
$\gamma>0$ is efficiently approximable.
There is an FPRAS for $\u3hTutte(2,\gamma)$.
\end{lemma}

\begin{proof}
 Jerrum and
Sinclair~\citeyear{JS93}
have given an FPRAS for $\Tutte(2,\gamma')$ for every $\gamma'>0$.
We will give a reduction from~$\u3hTutte(2,\gamma)$
to $\Tutte(2,\gamma')$ where $\gamma' = {(1+\gamma)}^{1/2}-1$.

Let $\hypergraph=(\hypervertices,\hyperedges)$ be a
$3$-uniform hypergraph, an instance of 
the source
problem 
$\u3hTutte(2,\gamma)$.
Let $\boldgamma$ be the constant function  with $\boldgamma_\hyperedge = \gamma$ for every $\hyperedge\in\hyperedges$.
Let $y=\gamma+1$.
Now, by Observation~\ref{obs:FK},
$$\ZTutte(\hypergraph;2,\boldgamma) = \ZPotts(\hypergraph;2,\boldgamma) = 
\sum_{\sigma:\hypervertices\rightarrow \{0,1\}}
y^{\mono(\sigma)},$$
where $\mono(\sigma)$ denotes the number of hyperedges $\hyperedge\in\hyperedges$
that are monochromatic in configuration~$\sigma$.

Construct a (multi-)graph~$\graph$ with vertex set~$\hypervertices$ 
and edge set
$$\graphedges = \bigcup_{(u,v,w)\in \hyperedges}
\{(u,v),(v,w),(u,w)\}.$$  
Let $\boldgamma'$ be the constant function  with $\boldgamma'_\edge = \gamma'$ for every $\edge\in\graphedges$.
Let $y'=\gamma'+1 = y^{1/2}$. 
Now if a hyperedge $\hyperedge\in \hyperedges$ is monochromatic in~$\sigma$, it contributes
${y'}^3$ to the corresponding term in $\ZPotts(\graph;2,\boldgamma')$.
Otherwise, it contributes $y'$ to the term.
Thus, $$\ZTutte(\graph;2,\boldgamma)=
\ZPotts(\graph;2,\boldgamma') = {y'}^{|\hyperedges|}\ZPotts(\hypergraph;2,\boldgamma),$$
which completes the proof.
\end{proof}

\bibliographystyle{plain}
\bibliography{\jobname}


\begin{thebibliography}{10}

\bibitem{dense}
Noga Alon, Alan Frieze, and Dominic Welsh.
\newblock Polynomial time randomized approximation schemes for
  {T}utte-{G}r\"othendieck invariants: the dense case.
\newblock {\em Random Structures Algorithms}, 6(4):459--478, 1995.

\bibitem{BGJ}
B.~Bollob{\'a}s, G.~Grimmett, and S.~Janson.
\newblock The random-cluster model on the complete graph.
\newblock {\em Probab. Theory Related Fields}, 104(3):283--317, 1996.

\bibitem{bordewich}
Magnus Bordewich.
\newblock On the approximation complexity hierarchy.
\newblock In Klaus Jansen and Roberto Solis-Oba, editors, {\em Approximation
  and Online Algorithms}, volume 6534 of {\em Lecture Notes in Computer
  Science}, pages 37--46. Springer Berlin / Heidelberg, 2011.

\bibitem{torpid}
Christian Borgs, Jennifer~T. Chayex, Jeong~Han Kim, Alan Frieze, Prasad Tetali,
  Eric Vigoda, and Van~Ha Vu.
\newblock Torpid mixing of some {M}onte {C}arlo {M}arkov chain algorithms in
  statistical physics.
\newblock In {\em Proceedings of the 40th Annual Symposium on Foundations of
  Computer Science}, FOCS '99, pages 218--229, Washington, DC, USA, 1999. IEEE
  Computer Society.

\bibitem{CSS}
Sergio Caracciolo, Alan~D. Sokal, and Andrea Sportiello.
\newblock Grassmann integral representation for spanning hyperforests.
\newblock {\em J. Phys. A}, 40(46):13799--13835, 2007.

\bibitem{Datalog}
S.~Ceri, G.~Gottlob, and L.~Tanca.
\newblock What you always wanted to know about {D}atalog (and never dared to
  ask).
\newblock {\em IEEE Trans. on Knowl. and Data Eng.}, 1(1):146--166, March 1989.

\bibitem{stablematchings}
Prasad Chebolu, Leslie~Ann Goldberg, and Russell Martin.
\newblock The complexity of approximately counting stable matchings.
\newblock {\em Theoretical Computer Science}, 437:35--68, 2012.

\bibitem{Dalmau05}
V\'{\i}ctor Dalmau.
\newblock Linear {D}atalog and bounded path duality of relational structures.
\newblock {\em Logical Methods in Computer Science}, 1(1:5), 2005.

\bibitem{APred}
Martin~E. Dyer, Leslie~Ann Goldberg, Catherine~S. Greenhill, and Mark Jerrum.
\newblock The relative complexity of approximate counting problems.
\newblock {\em Algorithmica}, 38(3):471--500, 2003.

\bibitem{BISpoly}
Qi~Ge and Daniel \v{S}tefankovi\v{c}.
\newblock {A graph polynomial for independent sets of bipartite graphs}.
\newblock In Kamal Lodaya and Meena Mahajan, editors, {\em IARCS Annual
  Conference on Foundations of Software Technology and Theoretical Computer
  Science (FSTTCS 2010)}, volume~8 of {\em Leibniz International Proceedings in
  Informatics (LIPIcs)}, pages 240--250, 2010.

\bibitem{confversion}
Leslie Goldberg and Mark Jerrum.
\newblock Approximating the partition function of the ferromagnetic {P}otts
  model.
\newblock In Samson Abramsky, Cyril Gavoille, Claude Kirchner, Friedhelm Meyer
  auf~der Heide, and Paul Spirakis, editors, {\em Automata, Languages and
  Programming}, volume 6198 of {\em Lecture Notes in Computer Science}, pages
  396--407. Springer Berlin / Heidelberg, 2010.

\bibitem{BISpolyslow}
Leslie Goldberg and Mark Jerrum.
\newblock A counterexample to rapid mixing of the {G}e-{S}tefankovic process.
\newblock {\em Electron. Commun. Probab.}, 17:no. 5, 1--6, 2012.

\bibitem{ising}
Leslie~Ann Goldberg and Mark Jerrum.
\newblock The complexity of ferromagnetic {I}sing with local fields.
\newblock {\em Combinatorics, Probability {\&} Computing}, 16(1):43--61, 2007.

\bibitem{tuttepaper}
Leslie~Ann Goldberg and Mark Jerrum.
\newblock Inapproximability of the {T}utte polynomial.
\newblock {\em Inform. and Comput.}, 206(7):908--929, 2008.

\bibitem{planartutte}
Leslie~Ann Goldberg and Mark Jerrum.
\newblock Inapproximability of the {T}utte polynomial of a planar graph.
\newblock {\em Computational Complexity}, 2012.
\newblock To Appear.

\bibitem{GJ99}
Vivek~K. Gore and Mark~R. Jerrum.
\newblock The {S}wendsen-{W}ang process does not always mix rapidly.
\newblock {\em J. Statist. Phys.}, 97(1-2):67--86, 1999.

\bibitem{Grimmett}
Geoffrey Grimmett.
\newblock Potts models and random-cluster processes with many-body
  interactions.
\newblock {\em J. Statist. Phys.}, 75(1-2):67--121, 1994.

\bibitem{Holley}
Richard Holley.
\newblock Remarks on the {${\rm FKG}$} inequalities.
\newblock {\em Comm. Math. Phys.}, 36:227--231, 1974.

\bibitem{JacksonSokal}
Bill Jackson and Alan~D. Sokal.
\newblock Zero-free regions for multivariate {T}utte polynomials (alias
  {P}otts-model partition functions) of graphs and matroids.
\newblock {\em Journal of Combinatorial Theory, Series B}, 99(6):869--903,
  2009.

\bibitem{JVW90}
F.~Jaeger, D.~L. Vertigan, and D.~J.~A. Welsh.
\newblock On the computational complexity of the {J}ones and {T}utte
  polynomials.
\newblock {\em Math. Proc. Cambridge Philos. Soc.}, 108(1):35--53, 1990.

\bibitem{JLR}
Svante Janson, Tomasz {\L}uczak, and Andrzej Rucinski.
\newblock {\em Random graphs}.
\newblock Wiley-Interscience Series in Discrete Mathematics and Optimization.
  Wiley-Interscience, New York, 2000.

\bibitem{JS93}
Mark Jerrum and Alistair Sinclair.
\newblock Polynomial-time approximation algorithms for the {I}sing model.
\newblock {\em SIAM J. Comput.}, 22(5):1087--1116, 1993.

\bibitem{jvv}
Mark~R. Jerrum, Leslie~G. Valiant, and Vijay~V. Vazirani.
\newblock Random generation of combinatorial structures from a uniform
  distribution.
\newblock {\em Theoret. Comput. Sci.}, 43(2-3):169--188, 1986.

\bibitem{KelkPhD}
Steven Kelk.
\newblock {\em On the relative complexity of approximately counting
  $H$-colourings}.
\newblock PhD thesis, University of Warwick, University of Warwick, Coventry,
  UK, July 2004.

\bibitem{LuczakLuczak}
Malwina Luczak and Tomasz {\L}uczak.
\newblock The phase transition in the cluster-scaled model of a random graph.
\newblock {\em Random Structures Algorithms}, 28(2):215--246, 2006.

\bibitem{Potts}
R.~B. Potts.
\newblock Some generalized order-disorder transformations.
\newblock {\em Proc. Cambridge Philos. Soc.}, 48:106--109, 1952.

\bibitem{Sly}
Allan Sly.
\newblock Computational transition at the uniqueness threshold.
\newblock In {\em Proceedings of the 2010 IEEE 51st Annual Symposium on
  Foundations of Computer Science}, FOCS '10, pages 287--296, Washington, DC,
  USA, 2010. IEEE Computer Society.

\bibitem{Sokal05}
Alan Sokal.
\newblock The multivariate {T}utte polynomial.
\newblock In {\em Surveys in Combinatorics}. Cambridge University Press, 2005.

\bibitem{Tutte84}
W.~T. Tutte.
\newblock {\em Graph theory}, volume~21 of {\em Encyclopedia of Mathematics and
  its Applications}.
\newblock Addison-Wesley Publishing Company Advanced Book Program, Reading, MA,
  1984.
\newblock With a foreword by C. St. J. A. Nash-Williams.

\bibitem{bipartite}
D.~L. Vertigan and D.~J.~A. Welsh.
\newblock The computational complexity of the {T}utte plane: the bipartite
  case.
\newblock {\em Combin. Probab. Comput.}, 1(2):181--187, 1992.

\bibitem{planar}
Dirk Vertigan.
\newblock The computational complexity of {T}utte invariants for planar graphs.
\newblock {\em SIAM J. Comput.}, 35(3):690--712 (electronic), 2005.

\bibitem{welsh}
D.~J.~A. Welsh.
\newblock {\em Complexity: knots, colourings and counting}, volume 186 of {\em
  London Mathematical Society Lecture Note Series}.
\newblock Cambridge University Press, Cambridge, 1993.

\bibitem{zuckerman}
David Zuckerman.
\newblock On unapproximable versions of {NP-Complete} problems.
\newblock {\em SIAM Journal on Computing}, 25(6):1293--1304, 1996.

\end{thebibliography}

\end{document}